\documentclass[12pt, article, onecolumn]{IEEEtran}
\IEEEoverridecommandlockouts                              
\usepackage{graphics} 
\usepackage{epsfig} 
\usepackage{subfloat}
\usepackage{subfigure}
\usepackage{times} 
\usepackage{amsmath} 
\usepackage{amssymb}  
\usepackage{mathrsfs}
\usepackage{cite}
\usepackage{epsf}
\usepackage{hyperref}
\usepackage{cleveref}
\usepackage{theorem}
\usepackage[active]{srcltx}

\newtheorem{theorem}{Theorem}
\newtheorem{proposition}{Proposition}
\newtheorem{lemma}{Lemma}

\newtheorem{remark}{Remark}

\newcommand{\ud}{\text{d}}
\newcommand{\Def}{\stackrel{\triangle}{=}}

\allowdisplaybreaks

\crefname{lemma}{Lemma}{Lemmas}
\Crefname{lemma}{Lemma}{Lemmas}
\crefname{theorem}{Theorem}{Theorems}
\Crefname{theorem}{Theorem}{Theorems}
\crefname{section}{Section}{Sections}
\Crefname{section}{Section}{Sections}
\crefname{figure}{Fig.}{Figs.}
\Crefname{figure}{Figure}{Figures}
\crefname{table}{Table}{Tables}
\Crefname{table}{Table}{Tables}
\crefname{equation}{Eq.}{Eqs.}
\Crefname{equation}{Equation}{Equations}
\crefname{proposition}{Proposition}{Propositions}
\Crefname{proposition}{Proposition}{Propositions}
\crefname{corollary}{Corollary}{Corollaries}
\Crefname{corollary}{Corollary}{Corollaries}
\crefname{problem}{Problem}{Problems}
\Crefname{problem}{Problem}{Problems}
\crefname{definition}{Definition}{Definitions}
\Crefname{definition}{Definition}{Definitions}

\ifCLASSINFOpdf
 
\else
\fi

\begin{document}
\title{Broadcast Approaches to the Diamond Channel}



\author{\IEEEauthorblockN{Mahdi Zamani and Amir K. Khandani}
\IEEEauthorblockA{\\ 
Department of Electrical and Computer Engineering\\
University of Waterloo, 
Waterloo, ON N2L 3G1\\
Emails: \{mzamani, khandani\}@uwaterloo.ca}
\thanks{This paper was presented in part at the IEEE International Symposium on
Information Theory, ISIT, Saint Petersburg, Russia, 2011.}
}


\maketitle

\begin{abstract}
%

The problem of dual-hop transmission from a source to a destination via two parallel full-duplex relays in block Rayleigh fading environment is investigated.
All nodes in the network are assumed to be oblivious to their forward channel gains; however, they have perfect information about their backward channel gains. 
We also assume a stringent decoding delay constraint of one fading block that makes the definition of ergodic (Shannon) capacity meaningless. 
The focus of this paper is on simple, efficient, and practical relaying schemes to increase the expected-rate at the destination. 
For this purpose, various combinations of relaying protocols and the broadcast approach (multi-layer coding) are proposed. 
For the decode-forward (DF) relaying, the maximum finite-layer expected-rate as well as two upper-bounds on the continuous-layer expected-rate are obtained.
The main feature of the proposed DF scheme is that the layers being decoded at both relays are added coherently at the destination although each relay has no information about the number of layers being successfully decoded by the other relay. It is proved that the optimal coding scheme is transmitting uncorrelated signals via the relays. 
Next, the maximum expected-rate of ON/OFF based amplify-forward (AF) relaying is analytically derived.
For further performance improvement, a hybrid decode-amplify-forward (DAF) relaying strategy, adopting the broadcast approach at the source and relays, is proposed and its maximum throughput and maximum finite-layer expected-rate are presented.
Moreover, the maximum throughput and maximum expected-rate in the compress-forward (CF) relaying adopting the broadcast approach, using optimal quantizers and Wyner-Ziv compression at the relays, are fully derived. 
All theoretical results are illustrated by numerical simulations. 
As it turns out from the results, when the ratio of the relay power to the source power is high, the CF relaying outperforms DAF (and hence outperforms both DF and AF relaying); otherwise, DAF scheme is superior. 

\end{abstract}
\IEEEpeerreviewmaketitle

\let\thefootnote\relax\footnotetext{Financial supports provided by Natural Sciences and Engineering Research Council of Canada (NSERC) and Ontario Ministry of Research \& Innovation (ORF-RE) are gratefully acknowledged.}

\section{Introduction} \label{intro}

The information theoretic aspects of wireless networks, have recently received wide attention. The widespread applications of wireless networks, along with many recent results in the network information theory area, have motivated efficient strategies for practical applications. Fading is often used for modeling the wireless channels \cite{biglieri}. The growing demand for quality of service (QoS) and network coverage inspires the use of several intermediate wireless nodes to help the communication among distant nodes, which is referred to as relaying or multi-hopping. Many papers analyze the information theoretic and communication aspects of relay networks. An information theoretic view of the three-node relay channel was proposed by Cover and El Gamal in \cite{cover-elgamal}, which was generalized in \cite{kramer} and \cite{xie} for multi-user and multi-relay networks. In \cite{cover-elgamal}, two different coding strategies were
introduced. In the first strategy, originally named ``cooperation"
and later known as ``decode-forward" (DF), the relay decodes the
transmitted message and cooperates with the source to send the
message in the next block. In the second strategy, ``compress-forward" (CF), the relay compresses the received
signal and sends it to the destination. Besides studying the DF and CF strategies, the authors in \cite{zahedi, schein, gastpar, laneman} have studied the
``amplify-forward" (AF) strategy for the Gaussian relay network.
In AF relaying, the relay amplifies and transmits its
received signal to the destination. Despite its simplicity, AF relaying  performs well in many
scenarios. El-Gamal and Zahedi \cite{zahedi} employed AF relaying in the single relay channel and derived the single
letter characterization of the maximum achievable rate using a simple linear
scheme (assuming frequency division and additive white Gaussian channel). 

The problems of transmission between a disconnected source and destination via two parallel intermediate nodes (the diamond channel) were analyzed in \cite{schein} for the additive white Gaussian channels and in \cite{sanderovich} for the case where the relays transmit in orthogonal frequency bands/time slots. There are also some asymptotic analyses on a source to destination communication via parallel relays with fading channels where the forward channels are known at both the transmitter and relays sides, see \cite{shahab} and references therein. Diversity gains in a parallel relay network using distributed space-time codes, where channel state information (CSI) is only at the receivers, was presented in \cite{hua}, \cite{jamshid} and references therein.

Here, we consider the problem of maximum expected-rate in the diamond channel. A good application for this network is a TV broadcasting system from a satellite to cellphones through base stations where users with better channels might receive additional services, such as high definition TV signal \cite{reza}. The growing adoption of broadcasting mobile TV services suggests that it has the potential to become a mass market application. However, the quality and success of such services are governed by guaranteeing a good coverage, particularly in areas that are densely populated. This paper suggests the use of relays to provide better coverage in such strategically important areas. The main transmitter which is a central TV broadcasting unit uses two parallel relays in each area with large density to improve coverage (see \cref{multicast_channel_model}). According to the large number of relay pairs covering their respective areas and also the large number of users in each designated area, neither the main transmitter nor the relays can access the forward channel state information. With no delay constraint, the ergodic nature of the fading channel can be experienced by sending very large transmission blocks, and the ergodic capacity is well studied \cite{biglieri}. According to the stringent delay constraint for the problem in consideration, the transmission block length is forced to be shorter than the dynamics of the slow fading process, though still large enough to yield a reliable communication. The performance of such channels are usually evaluated by outage capacity. The notion of capacity versus outage was introduced in \cite{ozarow} and \cite{biglieri}. Shamai and Steiner \cite{shamai} proposed a broadcast approach, a.k.a.\ multi-layer coding, for a single user block fading channel with no CSI at the transmitter, which maximizes the expected-rate. Since the expected-rate increases with the number of code layers \cite{shamai1997broadcast}, they evaluated the highest expected-rate using a continuous-layer (infinite-layer) code. This idea was applied to a dual-hop single-user channel in \cite{steiner} and a channel with two collocated cooperative users in \cite{steiner2}. The broadcast approach can also achieve the maximum average achievable rate in a block fading multiple-access channel with no CSI at the transmitters \cite{tse}. The optimized trade-off between the QoS and network coverage in a multicast network was derived in \cite{reza} using the broadcast approach.

\begin{figure}
\centering
\includegraphics[scale=0.3]{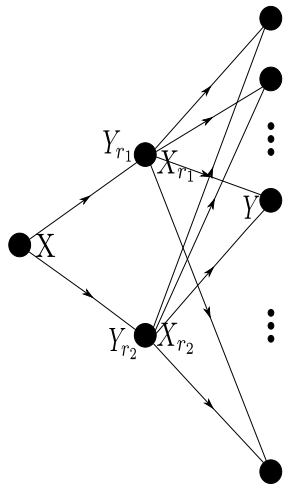}
\caption{Dual-hop multicast transmission via two parallel relays.}
\label{multicast_channel_model}
\end{figure}

In this paper, we investigate various relaying strategies in conjunction with the broadcast approach (multi-layer coding) scheme for the dual-hop channel with parallel relays where neither the source (main transmitter) nor the relays access the forward channels. Throughout the paper, we assume that channel gains are fixed during two consecutive blocks of transmission. The main focus of this paper is on simple and efficient schemes, since the relays can not buffer multiple packets and also handle large delays. Different relaying strategies such as DF, AF, hybrid DF-AF (DAF), and CF are considered. In DF relaying, a combination of the broadcast strategy and coding is proposed, such that the common layers, decoded at both relays, are decoded at the destination cooperatively. Note that each relay has no information about the number of layers being decoded by the other relay. The destination decodes from the first layer up to the layer that the channel condition allows. After decoding all common layers, the layers decodable at just one relay are decoded. It is proved that the optimal coding strategy is transmitting uncorrelated signals via the relays. Since the DF relaying in conjunction with continuous-layer coding is a seemingly intractable problem, 
the maximum finite-layer expected-rate is analyzed. Furthermore, two upper-bounds for the maximum continuous-layer expected-rate in DF are obtained.
In the DF relaying, the relays must know the codebook of the source and have enough time to decode the received signal. In the networks without these conditions, AF relaying is considered next. Both the maximum throughput and the maximum expected-rate, using a space-time code permutation between the relays, are derived. 
In the same direction and for further performance improvement, at the cost of increased complexity, a hybrid DF and AF scheme called DAF is proposed. In DAF with broadcast strategy, each relay decode-and-forwards a portion of the layers and amplify-and-forwards the rest. 
Afterwards, a multi-layer CF relaying is presented. 
In the CF relaying, the relays do not decode their received signals; instead, compress the signals by performing the optimal quantization in the Wyner-Ziv sense \cite{wyner}, which means each relay quantizes its received signal relying on the side information from the other relay.   
Besides the proposed achievable expected-rates, some upper bounds based on the channel enhancement idea and the max-flow min-cut theorem are obtained.
In all the proposed relaying strategies combined with the broadcast strategy, the maximum expected-rate increases with the number of code layers. It is numerically shown that when the ratio of the relay power to the source power is large, the CF relaying outperforms DAF, and hence outperforms both DF and AF; otherwise, DAF is the superior scheme. 

The rest of this paper is organized as follows: 
In \cref{system-model}, preliminaries are presented. Next, DF, AF, DAF, and CF relaying strategies in conjunction with the broadcast approach are elaborated in \cref{DF-multi-layer,AF-multi-layer,DAF-multi-layer,CF-multi-layer}, respectively. Afterwards, in \cref{upper-bound-multicast-relay}, some upper bounds on the maximum expected-rate are obtained. 
Numerical results are presented in \cref{numerical-results}. Finally, \cref{conclusion} concludes the paper.


\section{Preliminaries} \label{system-model}

\subsection{Notation}

Throughout the paper, we represent the expected operation by $\mathbb{E}(\cdot)$, the probability of event $A$ by $\Pr \{A\}$, the covariance matrix of random variables $X$ and $Y$ by $\mathbf R_{XY}$, the conditional covariance matrix of random variables $X$ and $Y$ by $\mathbf R_{XY|W,Z,\cdots}$, the differential entropy function by $\mathcal H(\cdot)$, and the mutual information function by $\mathcal I(\cdot;\cdot)$. The notation ``$\ln$'' is used for natural logarithm, and rates are expressed in \emph{nats}. We denote $f_{\mathrm x}(\cdot)$ and $F_{\mathrm x}(\cdot)$ as the probability density function (PDF) and the cumulative density function (CDF) of random variable $\mathrm x$, respectively. For every function $F(x)$, consider $\overline F(x)=1-F(x)$ and $F^{\prime}(x)=\frac{\ud F(x)}{\ud x}$. $\vec{X}$ is a vector and $\mathbf{Q}$ is a matrix. $\mathbf{I}_{n_t}$ denotes the $n_t \times n_t$ identity matrix. $s^o$ is the optimum solution with respect to the variable $s$. We denote the determinant, conjugation, matrix transpose, and matrix conjugate transpose operators by $\det$, $^*$, $^\text{T}$, and $^\dag$, respectively. $\mathcal U(\cdot)$ and $|\cdot|$ represent the unit step function and the absolute value or modulus operator, respectively. $\text{tr}(\mathbf{Q})$ denotes the trace of the matrix $\mathbf{Q}$. $\mathcal{CN}(0,1)$ denotes the complex Gaussian distribution with zero mean and unit variance. $\mathcal W$ is the Lambert $W$-function, also called the omega function, which is the inverse function of $f(W)=We^W$ \cite{corless}. $\text{E}_1(x)$ is the exponential integral function, which is $\int_x^\infty \frac{e^{-t}}{t} \ud t, x \geq 0$. $\Gamma(s,x)=\int_x^\infty t^{s-1} e^{-t} \ud t$ is the upper incomplete gamma function, and $\Gamma(s)=\Gamma(s,0)$. Throughout the paper, we assume that $\mathbb E \left( \left|X_i\right|^2 \right)=1,~ \forall i$.


\subsection{Network Model}

\begin{figure}
\centering
\includegraphics[scale=0.3]{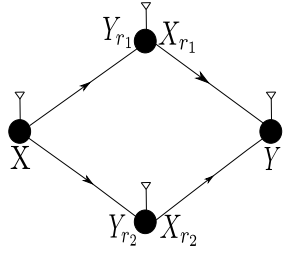}
\caption{Network model of dual-hop transmission from a single-antenna source to a single-antenna destination via two single-antenna relays, the diamond channel.}
\label{fig3_1}
\end{figure}

Let us first restate the network model. As \cref{fig3_1} shows, the destination receives data via two parallel relays and there is no direct link between the source and the destination.
The source transmits a signal $X$ subject to the total power constraint $P_{s}$, i.e., $\mathbb E\left( \left| X \right|^2 \right) \leq P_s$, and the received signal at the $\ell$'th relay is denoted by
\begin{align}
 Y_{r_\ell}=h_{r_\ell}X+Z_{r_\ell},~~~ \ell=1,2
\end{align}
The independent and identically distributed (i.i.d.) additive white Gaussian noise (AWGN) at the $\ell$'th relay is represented by $Z_{r_\ell} \sim \mathcal{CN}(0,1)$, and $h_{r_\ell} \sim \mathcal{CN}(0,1)$ is the channel coefficient from the source to the $\ell$'th relay. 
The $\ell$'th relay forwards a signal $X_{r_\ell}$ to the destination under the total power constraint $P_{r}$, i.e., $\mathbb E\left( \left| X_{r_\ell} \right|^2 \right) \leq P_r, ~ \ell=1,2$. The received signal at the destination is 
\begin{align}
Y=h_{1}X_{r_1}+h_{2}X_{r_2}+Z, 
\end{align}
where
$Z \sim \mathcal{CN}(0,1)$ is the i.i.d.\ AWGN and $h_{\ell} \sim \mathcal{CN}(0,1)$ is the channel coefficient from the $\ell$'th relay to the destination.
All $h_{r_\ell}$ and $h_{\ell}$ are assumed to be constant during two consecutive transmission blocks. Obviously, channel gains $a_{\ell}=|h_{\ell}|^2$ and $a_{r_\ell}=|h_{r_\ell}|^2$ have exponential distribution.

Note that the transmitter as well as both relays and the receiver are equipped with one antenna. 
We assume that the relays operate in a full-duplex mode and they are not capable of buffering data over multiple coding blocks or rescheduling tasks. Since there is no link between the relays, the half-duplex mode is a direct result of the full-duplex mode with frequency or time division.

\subsection{Definitions}

In the following, the performance metrics which are widely used throughout the paper are defined.
The expected-rate $\mathcal R_{f}$ is the average achievable rate when a multi-layer code is transmitted, i.e., the statistical expectation of the achievable rate. 
The maximum expected-rate, namely $\mathcal R_{f}^m$, is the maximum of the expected-rate over all transmit covariance matrices at the relays, transmission rates in each layer, and all power distributions of the layers. Mathematically,
\begin{align} \label{finite-layer-expected-rate-definition-formula}
\mathcal R_{f}^m  \Def  \max_{
\begin{subarray}{c}
R_i, P_i ,\mathbf Q_i \\
\text{tr}(\mathbf Q_i) \leq P_i \\
\sum_{i=1}^K P_i = P
\end{subarray}}
\sum_{i=1}^K \mathcal P_i\left( R_i \right) R_i,
\end{align}
where $R_i$, $\mathbf Q_i$, and $\mathcal P_i$ are the transmission rate, transmit covariance matrix at the relays, and probability of successful decoding in the $i$'th layer, respectively.

If a continuum of code layers are transmitted, the maximum continuous-layer (infinite-layer) expected-rate, namely $\mathcal R_{c}^m$, is given by maximizing the continuous-layer expected-rate over the layers' power distribution.

When a single-layer code is transmitted at the source and the relays, the average achievable rate is called the throughput, namely $\mathcal R_{s}$.
The maximum throughput, namely $\mathcal R_{s}^m$, is the maximum of the throughput 
over all transmit covariance matrices at the relays $\mathbf Q$, and transmission rates $R$. Mathematically,
\begin{align} \label{throughput-definition-formula}
\mathcal R_{s}^m  \Def  \max_{
\begin{subarray}{c}
R ,\mathbf Q \\
\text{tr}(\mathbf Q) \leq P
\end{subarray}}
 \mathcal P\left( R \right)  R.
\end{align}


\section{Decode-Forward Relays} \label{DF-multi-layer}

In order to enhance the lucidity of this section, single-layer coding is studied first. The idea is then extended to multi-layer coding. Since the continuous-layer expected-rate of this scheme is a seemingly intractable problem, a finite-layer coding scenario is analyzed in \cref{DF_ml}. 

\subsection{Maximum Throughput} \label{DF-outage}

In single-layer coding, a signal $X=\gamma X_1$ with power $P_{s}$ and rate $R=\ln(1+P_{s} s)$ is transmitted, where  $\gamma^2=P_{s}$. The $\ell$'th relay decodes and forwards the received signal in case $a_{r_\ell} \geq s$. If $a_{r_\ell}<s$, then $a_{r_\ell}$ is replaced by zero. The coding scheme at the relays is a distributed block space-time code in the Alamouti code sense \cite{alamouti}. 
At time $t$, the first relay sends $\alpha X_1(t)$ while the other relay sends $\beta X_1(t+1)$. To satisfy the relays power constraint, it is required that $\alpha^2=\beta^2=P_{r}$. At time $t+1$, the first and the second relays send $-\alpha X_1^*(t+1)$ and $\beta X_1^*(t)$, respectively. The relay with $a_{r_\ell}<s$ simply sends nothing. Applying the Alamouti decoding procedure and decomposing into two parallel channels, the throughput is given by
\begin{align} \label{pout-DF}
 \mathcal R_{D,s} =\Bigg[ \Pr \left\{ a_{r_1} \geq s \right\} \Pr \left\{ a_{r_2} \geq s \right\} \Pr \left\{ a_1+a_2 \geq s \frac{P_{s}}{P_{r}} \right\}+  \nonumber \\ 
\left. \Pr \left\{ a_{r_1} \geq s \right\} \Pr \left\{ a_{r_2} < s \right\} \Pr \left\{ a_{1} \geq s \frac{P_{s}}{P_{r}} \right\}+ \right. \nonumber \\ 
 \Pr \left\{ a_{r_1} < s \right\} \Pr \left\{ a_{r_2} \geq s \right\} \Pr \left\{ a_{2} \geq s \frac{P_{s}}{P_{r}} \right\} \Bigg] \ln(1+P_{s} s).
\end{align}
The first term in the right hand side of \cref{pout-DF} represents the case of decoding the signal at both relays and the destination. The second and third terms represent the probability of decoding the signal at only one relay and the destination.
Substituting the channel gain CDFs in (\ref{pout-DF}), the throughput is given by
\begin{align} \label{second_hop_alamouti_rate22}
 \mathcal R_{D,s} = \left(\!  \frac{P_{s}}{P_{r}} s e^{-s}\!-\!e^{-s}\!+\!2  \right)\! e^{-s \left(  \frac{P_{s}}{P_{r}} + 1 \right)} \ln(\!1\!+\!P_{s} s).
\end{align}

\Cref{DF-single-theorem} proves the optimality of the above scheme and presents the maximum throughput of the channel.

\begin{theorem} \label{DF-single-theorem}
In the proposed single-layer DF, 
the maximum throughput is achieved by sending uncorrelated signals on the relays.
the maximum throughput is given by
\begin{align} \label{second_hop_alamouti_rate22_2}
 \mathcal R_{D,s}^m\!=\!\!\max_{0 < s < s_t}\!\! \left(\!  \frac{P_{s}}{P_{r}} s e^{-s}\!-\!e^{-s}\!+\!2  \right)\! e^{-s \left(  \frac{P_{s}}{P_{r}} + 1 \right)} \ln(\!1\!+\!P_{s} s),
\end{align}
where $s_t = \min\left\{2\frac{P_r}{P_s},1.212\right\}$.
\end{theorem}

\begin{proof}

Consider $\mathbf Q\Def P_r \left[ \begin{matrix} 1 & \rho \\
\rho & 1 \end{matrix} \right]$ as the relays transmit covariance matrix. Therefore, $\mathbb E\left(X_{r_1}X_{r_2}^*\right)=\rho P_r$. 
In the following, we shall show that $\rho^o=0$.
Let us define $\overline F(s)$ as follows
\begin{align} \label{pout111}
 \overline F(s)\Def \Pr \left\{ a_{r_1} \geq s \right\} \Pr \left\{ a_{r_2} \geq s \right\} \Pr \left\{ a \geq s \frac{P_{s}}{P_{r}} \right\}+  \nonumber \\ 
\left. \Pr \left\{ a_{r_1} \geq s \right\} \Pr \left\{ a_{r_2} < s \right\} \Pr \left\{ a_{1} \geq s \frac{P_{s}}{P_{r}} \right\}+ \right. \nonumber \\ 
 \Pr \left\{ a_{r_1} < s \right\} \Pr \left\{ a_{r_2} \geq s \right\} \Pr \left\{ a_{2} \geq s \frac{P_{s}}{P_{r}} \right\},
\end{align}
where $a=\frac{1}{P_r} \vec h\mathbf Q\vec h^\dag$ and $\vec h = \left[\begin{matrix} h_1 & h_2 \end{matrix}\right]$.
The maximum throughput of the diamond channel in general form is 
\begin{align} \label{pout112}
 \mathcal R_{D,s}^m=\max_s \overline F(s) \ln(1+P_{s} s).
\end{align}
The only term in $\overline F(s) $ which depends on $\rho$ is $\Pr \left( a \geq s \frac{P_{s}}{P_{r}} \right)$.
Since $\mathbf Q$ is non-negative definite, one can write it as $\mathbf Q=\mathbf U\mathbf D\mathbf U^\dag$, where $\mathbf D=P_r \left[ \begin{matrix}
1+\rho & 0 \\
0 & 1-\rho \end{matrix} \right]$ is non-negative diagonal and $\mathbf U=\frac{1}{\sqrt{2}}\left[ \begin{matrix}
1 & 1 \\
1 & -1 \end{matrix} \right]$ is unitary. 
Since $h_1$ and $h_2$ are independent complex Gaussian random variables, each with independent zero-mean and equal variance real and imaginary parts, the distribution of $\vec h\mathbf U$ is the same as that of $\vec h$ \cite{telatar}. Thus,
\begin{align} \label{alamouti_proof_R3}
   \Pr \left\{ a \geq s \frac{P_{s}}{P_{r}} \right\} &= \Pr\left\{\vec h \mathbf Q 
   \vec h^\dag \geq P_{s} s \right\} \nonumber \\
   &= \Pr\left\{\left( \vec h \mathbf U \right) \mathbf D 
   \left(\vec h \mathbf U \right)^\dag \geq P_{s} s \right\} \nonumber \\
   &=\Pr\left\{\vec h \mathbf D 
   \vec h^\dag \geq P_{s} s \right\}.
\end{align}

The last expression in \cref{alamouti_proof_R3} corresponds to the complementary CDF in MISO channels.
Jorswieck and Boch \cite{jorswieck} proved that in an uncorrelated MISO channel with no CSI at the transmitter, but perfect
CSI at the receiver, for every transmission rate, the
optimal transmit strategy minimizing the outage probability is to use a fraction of all available
transmit antennas and perform equal power allocation with uncorrelated signals. Therefore, the solution of $\max_{\rho, \text{tr}(\mathbf D) \geq 2P_r}  \Pr\left\{\ln \left( 1+\vec h'
   \mathbf D
   \vec h^\dag \right) \geq \ln(1+P_{s} s) \right\}$ is $\rho=0$ or $\rho=1$. 


Defining 
\begin{align}
 s_c\Def - \left( 2\mathcal W_{-1} \left(\frac{-1}{2\sqrt{e}}\right)+1\right) \frac{P_r}{P_s} \approx 2.5129\frac{P_r}{P_s},
\end{align}
where $\mathcal W_{-1} \left(\cdot\right)$ is the -1 branch of the Lambert W-function, 
one can show that if $s \leq s_c$, then
\begin{align} \label{two_dist_single}
 \overline F_{\rho=0}(s) \geq \overline F_{\rho=1}(s). 
\end{align} 
In the remainder of the proof, we shall show that in case $\rho=1$, $s^o \leq s_c$. Then, as $\forall s \leq s_c$, $\overline F_{\rho=0}(s^o) \geq \overline F_{\rho=1}(s^o)$, it implies $\rho^o=0$, i.e., the optimum correlation coefficient between the relay signals maximizing the throughput of DF diamond channel is zero.

Assume that $s^o$ maximizes $\mathcal R(s)=\overline F_{\rho=1}(s) \ln \left( 1+P_s s  \right)$. Hence, ${\mathcal R^{\prime}}(s^o)=0$. Defining $f_{\rho=1}(s)=-\overline F_{\rho=1}^{\prime}(s)$, we get
\begin{align}
{\mathcal R^{\prime}}(s) = \overline F_{\rho=1} (s) \frac{P_s}{1+P_s s} - f_{\rho=1}(s) \ln \left( 1+P_s s  \right).
\end{align}

Let us define $g\left(s, P_s\right)=\ln \left( 1+P_s s \right)^{\frac{1+P_s s }{P_s}}$ and $r(s)=\frac{\overline F_{\rho=1}(s)}{f_{\rho=1}(s)}$. 
As such, we get
\begin{align} \label{diff_compair_trg}
\left\{
\begin{array}{lcl}
{\mathcal R^{\prime}}(s) > 0  & \text{iff} & r(s)>g\left(s, P_s\right), \\
{\mathcal R^{\prime}}(s) = 0  & \text{iff} & r(s)=g\left(s, P_s\right), \\
{\mathcal R^{\prime}}(s) < 0  & \text{iff} & r(s)<g\left(s, P_s\right). 
\end{array}
\right. 
\end{align}

Noting $\overline F_{\rho=1}(s)=\left( e^{-s}+2\left( 1-e^{-s} \right)e^{-\frac{P_s }{P_r}\frac{s}{2}} \right)e^{-\left( 1+\frac{P_s }{2P_r} \right)s}$, we have
\begin{align} 
r(s)\!=\!\frac{e^{-s}+2(1-e^{-s})e^{-s\frac{P_s}{2P_r}}}{\!\left(\! 2\!\!+\!\!\frac{P_s}{2P_r} \!\right)\!e^{-s}\!\!+\!\!2\left(\!1\!+\!\frac{P_s}{P_r}\!\right)\!e^{-s\frac{P_s}{2P_r}}\!\!-\!\!2\!\left(\!2\!+\!\frac{P_s}{P_r}\!\right)\!e^{-s}e^{-s\frac{P_s}{2P_r}}}.
\end{align}
It can be shown that as far as $s\geq s_t = \min\left\{2\frac{P_r}{P_s},1.212\right\}$, we have  
\begin{align} \label{lhs_diff_miso_single_Rayleigh3} 
r(s) < s, ~~~\forall s\geq s_t.
\end{align}  
The derivative of $g\left(s, P_s\right)$ over $P_s$ is 
\begin{align}
\frac{\partial g\left(s, P_s\right)}{\partial P_s} \!=\! \frac{s P_s\!-\!\ln\!\left(\! 1\!\!+\!\! s P_s \! \right)}{P_s^2} \!=\!
\frac{1}{P_s^2}\!\ln\!\! \left(\!\! 1\!+\!\frac{1}{1\!\!+\!\!s P_s}\!\! \sum_{k=2}^{\infty} \!\!\frac{\left(s P_s \right)^k}{k!} \!\right)\!\!\! \geq \!\!0.
\end{align}
Therefore, $g\left(s, P_s\right)$ is a monotonically increasing function of $P_s$ and its minimum is in $P_s=0$. As a result,
\begin{align}  \label{lhs_diff_miso_single_Rayleigh5}
 g\left(s, P_s\right) > \lim_{P_s \rightarrow 0} \ln \left( 1+ P_s s\right)^{\frac{1+P_s s}{P_s}} = s.
\end{align}
Comparing \cref{lhs_diff_miso_single_Rayleigh3}, \cref{lhs_diff_miso_single_Rayleigh5}, $r(0)=\frac{2P_r}{P_s}>0$ and $g\left(0, P_s\right)=0$ yields
\begin{align} \label{lhs_diff_miso_single_Rayleigh6}
\left\{
\begin{array}{ll}
r(s)>g\left(s, P_s\right) & s=0, \\
r(s)<g\left(s, P_s\right) & s\geq s_t.
\end{array}
\right. 
\end{align}
Applying \cref{lhs_diff_miso_single_Rayleigh6} to \cref{diff_compair_trg} gives
\begin{align} \label{diff_miso_single_Rayleigh1} 
\left\{
\begin{array}{ll}
{\mathcal R^{\prime}}(s)>0 & s=0, \\
{\mathcal R^{\prime}}(s)<0 & s\geq s_t.
\end{array}
\right. 
\end{align}

As $\mathcal R(s)$ is a continuous function, according to \cref{diff_miso_single_Rayleigh1}, $0< s^o < s_t$. Noting $s_t<s_c$, \cref{two_dist_single} yields $\overline F_{\rho=0}(s^o) > \overline F_{\rho=1}(s^o)$ and as a result, $\rho^o=0$ and $a=a_1+a_2$. Substituting the channel gain CDFs in \cref{pout111}, the maximum throughput of the DF diamond channel is given by \cref{second_hop_alamouti_rate22_2}, which is achievable by applying the aforementioned distributed space-time code.

\end{proof}

\subsection{Maximum Finite-Layer Expected-Rate} \label{DF_ml}

For the lucidity of this section, the encoding and decoding procedures are presented sparately.


\subsubsection{Encoding Procedure}
The transmitter sends a $K$-layer code $X = \sum_{i=1}^K \gamma_{i} X_{i}$ to the relays, where $\gamma_{i}^2$ represents the power allocated to the $i$'th layer with rate 
\begin{align} \label{rate_layer_3}
R_{i}=\ln \left( 1+\frac{\gamma_{i}^2 s_i}{1+\sum_{j=i+1}^K \gamma_{j}^2 s_i} \right).
\end{align}
The relays start decoding the received signal from the first layer up to the layer that their backward channel conditions allow. Then, the relays re-encode and forward the decoded layers to the destination. 
To design the transmission strategy, we first state \cref{df-multi-strategy-theorem}.

\begin{theorem} \label{df-multi-strategy-theorem}
In multi-layer DF, if the layers' power distribution in the first relay is equal to that of the second relay, the relay signals must be uncorrelated in order to achieve the maximum expected-rate.
\end{theorem}

\begin{proof}

Analogous to the proof of \cref{DF-single-theorem}, let us define 
\begin{align} \label{pmult111}
 \mathcal P_i &\Def \overline F_{\mathrm a_{r_1}}(s_i) \overline F_{\mathrm a_{r_2}}(s_i) \mathcal P_{i,1,2} \nonumber \\ 
&+ \overline F_{\mathrm a_{r_1}}(s_i) F_{\mathrm a_{r_2}}(s_i) \mathcal P_{i,1} 
+ F_{\mathrm a_{r_1}}(s_i) \overline F_{\mathrm a_{r_2}}(s_i) \mathcal P_{i,2},
\end{align}
where $\mathcal P_{i,1,2}$, $\mathcal P_{i,1}$, and $\mathcal P_{i,2}$ are the probability of decoding the $i$'th layer at the destination when both relays, only the first relay, and only the second relay decode the signal, respectively.
The expected-rate in the $i$'th layer can be written as
\begin{align} \label{pmult12}
 \mathcal R_i(s) = \mathcal P_i \ln \left( 1+\frac{\gamma_{i}^2 s_i}{1+\sum_{j=i+1}^K \gamma_{j}^2 s_i } \right).
\end{align}

The only term in \cref{pmult111} which depends on the transmit strategy at the relays is $\mathcal P_{i,1,2}$. We denote $\mathbf Q_i$ as the transmit covariance matrix of the relays in the $i$'th layer. So that,
\begin{align}
\mathcal P_{i,1,2} = \Pr \left\{ 1+\frac{\vec h \mathbf Q_i \vec h^\dag}{1+\vec h \sum_{j=i+1}^K \mathbf Q_j \vec h^\dag} \geq R_i  \right\}.
\end{align}
Analogous to the proof of \cref{DF-single-theorem}, by decomposing $\mathbf Q_i$ and $\sum_{j=i+1}^K \mathbf Q_j$, and noting the fact that multiplying $\vec h$ by any unitary matrix does not change the distribution of $\vec h$, we get
\begin{align} \label{multi-df-th-pr-1}
\mathcal P_{i,1,2} = \Pr \left\{ 1+\frac{P_i \vec h \begin{bmatrix} 1+\rho_i & 0 \\ 0 & 1-\rho_i \end{bmatrix} \vec h^\dag}{1+I_i \vec h \begin{bmatrix} 1+\hat \rho_i & 0 \\ 0 & 1-\hat \rho_i \end{bmatrix} \vec h^\dag} \geq R_i  \right\}.
\end{align}
It can be shown that the optimum solutions for $\rho$ and $\hat \rho$ to minimize $\mathcal P_{i,1,2}$ in \cref{multi-df-th-pr-1} is either $\rho_i=\hat \rho_i=0$ or $\rho_i=\hat \rho_i=1$ \cite{zamani2012maximum}. We shall now show that the optimum solution is $\rho_i^o=\hat \rho_i^o=0$. Towards this, we follow the same general outline to the proof of \cref{DF-single-theorem}.

Let us define the following functions,
\begin{align}
&g\left( s_i,P_i,I_i \right)=\frac{\left(1+I_i s_i\right)\left(1+\left( I_i+P_i\right) s_i\right)}{P_i} \ln \left( 1+\frac{P_i 
s_{i}}{1+I_i s_i} \right), \\
&r(s_i)=-\frac{\mathcal P_i}{\frac{\ud \mathcal P_i}{\ud s_i}}.
\end{align}
One can simply show that \cref{diff_compair_trg,lhs_diff_miso_single_Rayleigh3} still hold by redefining the functions as above, and with $s$ replaced by $s_i$.
 
Defining $\hat P \Def \frac{P_i}{1+I_i s_i}$, from \cref{lhs_diff_miso_single_Rayleigh5} and noting $I_i s_i \geq 0$, we have
\begin{align}
g\left( s_i,P_i,I_i \right) &= \left( 1+I_i s_i \right) \frac{\left(1+\frac{P_i s_i
}{1+I_i s_i} \right)}{\frac{P_i}{1+I_i s_i}} \ln \left( 1+\frac{P_i s_i
}{1+I_i s_i} \right) \nonumber \\
& {\geq} \frac{\left(1+\frac{P_i s_i
}{1+I_i s_i} \right)}{\frac{P_i}{1+I_i s_i}} \ln \left( 1+\frac{P_i s_i
}{1+I_i s_i} \right) \nonumber \\
&=\ln \left( 1+\hat P s_i \right)^{\frac{\left(1+\hat P s_i \right)}{\hat P}}
> s_i.
\end{align}
Therefore, \cref{lhs_diff_miso_single_Rayleigh6,diff_miso_single_Rayleigh1} still hold with the above functions, and then, $0 < s_i^o < s_t$.
Noting $s_t<s_c$ results because as pointed out earlier $\mathcal P_{i,\rho=0}(s_i^o) > \mathcal P_{i,\rho=1}(s_i^o)$.

\end{proof}
With respect to \cref{df-multi-strategy-theorem}, the following transmission scheme is proposed. 
Assume that the first and the second relays decode $M$ and $N$ layers out of the whole $K$ transmitted layers, respectively, according to their corresponding backward channel. 
As the relays do not know the channel of the other relay, and hence, do not know the layers' power distribution in the other relay, its code construction is based on a similar power distribution assumption for the other relay.
\Cref{df-multi-strategy-theorem} demonstrates that uncorrelated signals must be transmitted over the relays. For this purpose, the following scheme is proposed. At time $t$, the first relay sends $\sum_{i=1}^K \alpha_{i} X_{i} (t)$ while the other relay sends $\sum_{i=1}^K \beta_{i} X_{i} (t+1)$. At time $t+1$, the first and the second relays send $\sum_{i=1}^K - \alpha_{i} X_{i}^{*} (t+1)$ and $\sum_{i=1}^K \beta_{i} X_{i}^{*} (t)$, respectively. Note that $\sum_{i=1}^M \alpha_i^2 = P_r$, $\alpha_{i}=0$ for $i=M+1,...,K$ and $\sum_{i=1}^N \beta_i^2 = P_r$, $\beta_{i}=0$ for $i=N+1,...,K$.

The received signal at the destination is
\begin{align} \label{alam1-DF}
\begin{cases}
Y(t)=h_{1}\sum_{i=1}^K\alpha_{i}X_{i}(t)+h_{2}\sum_{i=1}^K \beta_{i}X_{i} (t+1)+Z(t), \\
Y\!(\!t\!+\!1\!)\!=\!\!-h_{1}\!\sum_{i=1}^K\! \alpha_{i}X_{i}^{*}\!(\!t\!+\!1\!)\!+\!h_{2}\!\sum_{i=1}^K \!\beta_{i}X_{i}^{*}\! (t)\!+\!Z(\!t\!+\!1\!).
\end{cases}  
\end{align}
One may express a matrix representation for \cref{alam1-DF} as 
\begin{align} \label{alam4-DF}
 \begin{bmatrix}
  Y(t) \\
\!-Y^*(t\!+\!1)\!
\end{bmatrix}
\!\!=\!\!
\sum_{i=1}^K
\begin{bmatrix}
 h_{1}\alpha_{i} & h_{2}\beta_{i} \\
-h_{2}^{*}\beta_{i} & h_{1}^{*}\alpha_{i}
\end{bmatrix}
\!\! \begin{bmatrix}
  X_{i}(t) \\
\!X_{i}(t\!+\!1)\!
\end{bmatrix}
\!\!+\!\!
 \begin{bmatrix}
  Z(t) \\
\!-Z^{*}(t\!+\!1)\!
\end{bmatrix}.
\end{align}

\subsubsection{Decoding procedure} 

The destination starts decoding the code layers in order, from the first layer up to the highest layer that is decodable. To decode the $i$'th layer, after decoding the first $i-1$ layers, the channels are separated into two parallel channels by multiplying both sides of \cref{alam4-DF} by $\begin{bmatrix}
 h_{1}^{*}\alpha_{i} & -h_{2}\beta_{i} \\
h_{2}^{*}\beta_{i} & h_{1}\alpha_{i}
\end{bmatrix}$. Therefore,
\begin{align}
 &\begin{bmatrix}
  \tilde{Y}(t) \\
\tilde{Y}(t+1)
 \end{bmatrix}
= 
\begin{bmatrix}
 a_1\alpha_{i}^{2}+a_2\beta_{i}^{2} & 0 \\
0 & a_1\alpha_{i}^{2}+a_2\beta_{i}^{2}
\end{bmatrix}
\begin{bmatrix}
  X_{i}(t) \\
X_{i}(t+1)
 \end{bmatrix} +
\nonumber \\
&
\sum_{j=i+1}^K 
\!\!\!\begin{bmatrix}
 h_{1}^{*}\alpha_{i} & -h_{2}\beta_{i} \\
h_{2}^{*}\beta_{i} & h_{1}\alpha_{i}
\end{bmatrix}
\!\!\begin{bmatrix}
 h_{1}\alpha_{j} & h_{2}\beta_{j} \\
-h_{2}^{*}\beta_{j} & h_{1}^{*}\alpha_{j}
\end{bmatrix}
\!\!\begin{bmatrix}
  X_{j}(t) \\
\!X_{j}(t\!+\!1)\!
 \end{bmatrix} \!\!+\!\!
\begin{bmatrix}
  \tilde Z(t) \\
\!\tilde Z(t\!+\!1)\!
 \end{bmatrix}.
\end{align}
$\tilde Z(t)$ and $\tilde Z(t+1)$ are two independent i.i.d AWGN, each with power $a_1  \alpha_i^2+a_2 \beta_i^2$.



The interference power caused by upper layers while decoding the $i$'th layer is
\begin{align} \label{cdf_layer_DF_ml_1}
I_i&=
\sum_{j=i+1}^K \left( \left( a_1 \alpha_i \alpha_j + a_2 \beta_i \beta_j \right)^2 + a_1 a_2 \left( \alpha_i \beta_j - \alpha_j \beta_i  \right)^2 \right) \nonumber \\
&=\left( a_1 \alpha_i^2 +a_2 \beta_i^2 \right) \sum_{j=i+1}^K \left( a_1 \alpha_j^2 +a_2 \beta_j^2 \right).
\end{align}
Thus, the probability that the $i$'th layer can be successfully decoded at the destination is
{\small \begin{align}
\mathcal P_i=
\Pr \left\{ \frac{ a_1 \alpha_i^2 + a_2 \beta_i^2 }{1 + \sum_{j=i+1}^K \left(  a_1 \alpha_j^2 + a_2 \beta_j^2 \right)} \geq \frac{\gamma_{i}^2 s_i}{1+\sum_{j=i+1}^K \gamma_{j}^2 s_i} \right\}.
\end{align}}
Hence, the achievable expected-rate using this scheme can be written as
\begin{align}
 \mathcal R_{D,f}=
\sum_{i=0}^K \mathcal{P}_i \ln \left( 1+\frac{\gamma_{i}^2 s_i}{1+\sum_{j=i+1}^K \gamma_{j}^2 s_i} \right).
\end{align}
To summarize, we have shown the following.

\begin{theorem} 
In the diamond channel, the above result implies that the following expected-rate is achievable.
\begin{align}
 \mathcal R_{D,f}^m=\max_{\begin{subarray}{c}
s_i, \gamma_i, \alpha_i, \beta_i
\end{subarray}} \sum_{i=0}^K \mathcal P_i \ln \left( 1+\frac{\gamma_{i}^2 s_i}{1+\sum_{j=i+1}^K \gamma_{j}^2 s_i} \right),
\end{align}
with $\mathcal P_i=
\Pr \left\{ \frac{ \left| h_1 \right|^2 \alpha_i^2 + \left|h_2 \right|^2 \beta_i^2}{1+ \sum_{j=i+1}^K \left( \left| h_1 \right|^2 \alpha_j^2 + \left|h_2 \right|^2 \beta_j^2 \right)} \geq \frac{\gamma_{i}^2 s_i}{1+\sum_{j=i+1}^K \gamma_{j}^2 s_i} \right\}$.
The maximization is subject to $\sum_{i=1}^K \gamma_{i}^2=P_{s}$, $\sum_{i=1}^K \alpha_{i}^2=\sum_{i=1}^K \beta_{i}^2=P_{r}$, where $\alpha_i$ and $\beta_i$ are zero for the layers which are not decoded at the relays.
Note that $\alpha_i$s and $\beta_i$s are optimized separately.
\end{theorem}

\begin{remark} 
One important feature of the proposed scheme is that the layers being decoded at both relays are added coherently at the destination although each relay has no information about the number of layers being successfully decoded by the other relay.
\end{remark}



\section{Amplify-Forward Relays}   \label{AF-multi-layer}

A simple but efficient relaying
solution for the diamond channel is to amplify and forward the received
signals. In order for the destination to coherently decode the signals, it employs a distributed
space-time code permutation along with the threshold-based ON/OFF power scheme, which has been shown that improves the performance of AF relaying \cite{hua}. 
According to the ON/OFF concept, any relay whose backward channel gain is less
than a pre-determined threshold, namely $a_{th}$, is silent.
In this scheme, the relays transmit the signals to the destination in two consecutive time slots. In time slot $t$, the first (resp.\ second) relay transmits $c_1 Y_{r_1}(t)$ (resp.\ $c_2 Y_{r_2}(t+1)$). In time slot $t+1$, the first 
(resp.\ second) relay transmits $-c_1
Y_{r_1}^*(t+1)$ (resp.\ $c_2 Y_{r_2}^*(t)$) with the \emph{backward channel
phase compensation} \cite{hua}. To satisfy the relays' power constraint, it is required that $c_\ell =
\sqrt{\frac{\mathcal U\left(a_{r_\ell}-a_{th}\right) P_r}{a_{r_\ell} P_s+1}}$, $\ell=1, 2$, where $\mathcal U(\cdot)$ is the unit step function.
At the destination, the channels are parallelized using the Alamouti decoding procedure \cite{alamouti}. 
The received signal at the destination is 
\begin{align}
\begin{cases}
 Y(t){=} c_1 h_{1}(t)Y_{r_1}(t)+c_2 h_{2}(t)Y_{r_2}(t+1)+Z(t), \\
Y(t+1){=}\! -c_1 h_{1}(t)Y_{r_1}^{*}(t+1)\! +\!c_2 h_{2}(t)Y_{r_2}^*(t) \!+\!Z(t+1)
.
\end{cases}
\end{align}
As the destination accesses the backward channels, after compensating the phases of $h_{r_1}$ and $h_{r_2}$ into $h_{r_1}^*$ and $h_{r_2}^*$ in time slot $t+1$, we get
\begin{align} \label{AF-single-matrix-decoder}
 \begin{bmatrix}
  Y(t) \\
-Y^*(t+1)
 \end{bmatrix}
=
\begin{bmatrix}
 h_{r_1}h_{1}c_1 & h_{r_2}h_{2}c_2 \\
-h_{r_2}^*h_{2}^*c_2 & h_{r_1}^*h_{1}^*c_1
\end{bmatrix}
\begin{bmatrix}
  X(t) \\
X(t+1)
 \end{bmatrix}
 \nonumber \\
+
\begin{bmatrix}
 c_1 h_{1} Z_{r_1}(t)+c_2 h_{2}Z_{r_2}(t+1) +Z(t)    \\
c_1 h_{1}^* Z_{r_1}(t+1) -c_2 h_{2}^* Z_{r_2}(t) -Z^*(t+1)
\end{bmatrix}.
\end{align}

Multiplying $\begin{bmatrix}
 h_{r_1}h_{1}c_1 & h_{r_2}h_{2}c_2 \\
-h_{r_2}^*h_{2}^*c_2 & h_{r_1}^*h_{1}^*c_1
\end{bmatrix}^\dag
$ to both sides of \cref{AF-single-matrix-decoder}, two channels are parallelized, and the source-destination instantaneous mutual information is
\begin{align}
 \mathcal I\left(X ; Y\right) = \ln \left( 1+\frac{\left( |h_{r_1}h_{1}|^2c_1^2+|h_{r_2}h_{2}|^2 c_2^2 \right) P_s }{1+|h_{1}|^2c_1^2+|h_{2}|^2c_2^2} \right),
\end{align}
which is equivalent to a point-to-point channel with the following channel gain,
\begin{align} \label{AF-equivalent-2-channel-gain}
a_{AF,2}\Def \frac{\frac{P_r}{a_{r_1} P_{s}+1}a_{r_1}a_{1}+\frac{P_r}{a_{r_2} P_{s}+1}a_{r_2}a_{2}}{1+\frac{P_r}{a_{r_1} P_{s}+1}a_{1}+\frac{P_r}{a_{r_2} P_{s}+1}a_{2}}.
\end{align}
If one relay is silent and only one relay transmits, let say the $\ell$'th relay, by replacing zero instead of one of the channel gains into \cref{AF-equivalent-2-channel-gain}, we get
\begin{align} \label{AF-equivalent-1-channel-gain}
a_{AF,1}\Def \frac{a_{r_\ell}a_{\ell} P_{r}}{1+a_{r_\ell} P_s+a_{\ell} P_r}.
\end{align}
The expected value of the optimum ON/OFF threshold in which $a_{AF,2}>a_{AF,1}$ is given by
\begin{align} \label{optimum-threshold-AF}
 a_{th} = \frac{P_r}{1+P_s+P_r}.
\end{align}


\Cref{AF_throughput} yields the maximum achievable throughput in this method. 

\begin{proposition} \label{AF_throughput}
The maximum achievable throughput in the above AF scheme is specified by
\begin{align}
 \mathcal R_{A,s}^{m}= \max_{s} \,\,\, & e^{-\frac{P_r}{1+P_s+P_r}}  \left( e^{-\frac{P_r}{1+P_s+P_r}} \overline F_{\mathrm a_{AF,2}}(s) \right. \nonumber \\ 
 &\left. + 2\left( 1 - e^{-\frac{P_r}{1+P_s+P_r}} \right) \overline F_{\mathrm a_{AF,1}}(s) \right) \ln (1+{P}_{s} s),
\end{align}
where 
$F_{\mathrm a_{AF,2}}(\cdot)$ and $F_{\mathrm a_{AF,1}}(\cdot)$ are the CDFs of $a_{AF,2}$ and $a_{AF,1}$ from \cref{AF-equivalent-2-channel-gain,AF-equivalent-1-channel-gain}, respectively.
\end{proposition}
The maximum continuous-layer expected-rate of the above AF relaying is presented in \cref{AF_expected}.

\begin{theorem} \label{AF_expected}
{The maximum achievable expected-rate in the above AF relaying is given by
\begin{align}
 \mathcal R_{A,c}^m{=}e^{-\frac{P_r}{1+P_s+P_r}}\!\!\left(\! 2-e^{-\frac{P_r}{1+P_s+P_r}}\! \right)\!\! \int\limits_{s_{0}}^{s_{1}} \! \overline F(s) \left(\! \frac{2}{s}\! +\! \frac{ f^{\prime}(s)}{f(s)}\! \right) \emph{d} s,
\end{align}
with
\begin{align} \label{feq3}
 F(s)&\Def \frac{2\left( e^{\frac{P_r}{1+P_s+P_r}}-1\right) F_{\mathrm a_{AF,1}}(s)+ F_{\mathrm a_{AF,2}}(s) }{2e^{\frac{P_r}{1+P_s+P_r}}-1}, \\
f(s) &\Def F^{\prime}(s).
\end{align}
The integration limits are the solutions to $\overline F(s_0)=s_0(1+P_s s_0)f(s_0)$ and $\overline F(s_1)=s_1f(s_1)$, respectively.
}
\end{theorem}

\begin{proof}

The maximum achievable expected-rate at the destination can be expressed by
\begin{align} \label{on-off1}
 \mathcal R_{A,c}^m&{=}2 e^{-a_{th}} \left( 1-e^{-a_{th}}  \right) \mathcal R_1^m+  e^{-2a_{th}} \mathcal R_{2}^m ~~~~~ \nonumber \\
&{=} e^{-a_{th}}\! \left( 2-e^{-a_{th}} \! \right)\!\! \left(\! \frac{2\left(1\!-\!e^{-a_{th}}\right)}{2\!-\!e^{-a_{th}}} \mathcal R_{1}^m\! +\! \frac{e^{-a_{th}}}{2\!-\!e^{-a_{th}}} \mathcal R_{2}^m \! \right),
\end{align}
where $\mathcal R_{1}^m$ and $\mathcal R_{2}^m$ are the maximum expected-rates when only one relay is active and both relays are active, respectively. 
As showed in \cite{shamai,steiner2007multi}, $\mathcal R_1^m$ and $\mathcal R_2^m$ are given by 
\begin{align} \label{infinite-layer-general-shamai_2}
\mathcal R_1^m = \max_{I(s)} \int_{0}^{\infty} \overline F_{\mathrm a_{AF,1}}(s) \frac{-s I^{\prime}(s)}{1+sI(s)} \ud s, \nonumber \\
\mathcal R_2^m = \max_{I(s)} \int_{0}^{\infty} \overline F_{\mathrm a_{AF,2}}(s) \frac{-s I^{\prime}(s)}{1+sI(s)} \ud s.
\end{align}

Substituting the above equations in \cref{on-off1}, we get
\begin{align}\label{on-off4}
\mathcal R_{A,c}^m
&{=}\max_{I(s)} e^{-a_{th}}\!\left( 2-e^{-a_{th}}  \right)\!\! \int\limits_{0}^{\infty} \!\! \left( 1-\frac{2\left(1-e^{-a_{th}}\right)}{2-e^{-a_{th}}} F_{\mathrm a_{AF,1}}(s)\right.   \nonumber \\
&- \left. \frac{e^{-a_{th}}}{2-e^{-a_{th}}} F_{\mathrm a_{AF,2}}(s) \right) \frac{-x I^{\prime}(s)}{1+sI(s)} \ud s.
\end{align}
Defining 
\begin{align}
F(s)\Def \frac{2\left(1-e^{-a_{th}}\right)}{2-e^{-a_{th}}} F_{\mathrm a_{AF,1}}(s) + \frac{e^{-a_{th}}}{2-e^{-a_{th}}}  F_{\mathrm a_{AF,2}}(s),
\end{align}
the maximum expected-rate of the proposed AF scheme is found by
\begin{align}\label{on-off6}
\mathcal R_{A,c}^m= \max_{I(s)} e^{-a_{th}}\left( 2-e^{-a_{th}}\right) \int_{0}^{\infty} \overline F(s) \frac{-s I^{\prime}(s)}{1+sI(s)} \ud s.
\end{align}
Substituting $a_{th}$ by $\frac{P_r}{1+P_s+P_r}$ and maximizing over $I(s)$ by solving the corresponding E$\ddot{\text{u}}$ler equation \cite{calculus}, we come up with the maximum expected-rate as
\begin{align}\label{on-off7}
\mathcal R_{A,c}^m{=}e^{-\frac{P_r}{1+P_s+P_r}}\!\!\left(\! 2-e^{-\frac{P_r}{1+P_s+P_r}}\! \right)\!\!\! \int\limits_{s_{0}}^{s_{1}} \!\! \overline F(s)\!\! \left(\! \frac{2}{s}\! +\! \frac{ f^{\prime}(s)}{f(s)}\! \right)\!\! \ud s,
\end{align}
where $s_0$ and $s_1$ are the solutions to $\overline F(s_0)=s_0(1+P_s s_0)f(s_0)$ and $\overline F(s_1)=s_1 f(s_1)$, respectively.

\end{proof}


\begin{remark} 
In the above results, the power constraint $P_r$ has been applied only to the time slots when the relays are ON. Alternatively, one can assume that the relays have the ability to save their power while working in the OFF state and consume it in the ON state. In this case, all the above calculations hold except for the integration limit $s_0$ which is now the solution to $\overline F(s_0)=s_0(1+e^{\frac{P_r}{1+P_s+P_r}} P_s s_0)f(s_0)$.
\end{remark}





\section{Hybrid Decode-Amplify-Forward Relays} \label{DAF-multi-layer}

In this section, we propose a DAF relaying strategy which takes advantage of amplifying the layers that could not be
decoded at the relays in the DF scheme. Specifically, each relay tries to decode as many layers as
possible and forward them by spending a portion of its power budget. The remaining power is dedicated to amplifying and forwarding the rest of the layers. 

In order to enhance the lucidity of this section, single-layer coding is studied first. The idea is then extended to multi-layer coding. As the continuous-layer expected-rate of this scheme is a seemingly intractable problem, a finite-layer coding scenario is analyzed.

\subsection{Maximum Throughput}

A single-layer code $X = \gamma X_1$ with power ${P}_{s}$, i.e., $\gamma^2={P}_{s}$,
and rate $R=\ln(1+ {P}_{s} s)$ is transmitted. If $a_{r_\ell} \geq s$, then the $\ell$'th relay decodes the signal and forwards it, otherwise, it amplifies and forwards the received signal to the destination. In time slot $t$, the first (resp.\ second)
relay transmits $X_{r_1}(t)$ (resp.\ $X_{r_2}(t+1)$).
In time slot $t+1$, the first (resp.\ second) relay transmits
$-X_{r_1}^*(t+1)$ (resp.\ $X_{r_2}^*(t)$) with the \emph{backward channel phase compensation}. There are three possibilities:
\begin{enumerate}
\item $a_{r_1} \geq s$ and $a_{r_2} \geq s$: both relays decode the signal. In this case DAF is simplified to DF in \cref{DF-multi-layer}.
\item $a_{r_1} < s$ and $a_{r_2} < s$: none of the relays decodes the signal. This case is simplified to AF in \cref{AF-multi-layer}.
\item $a_{r_1} \geq s , a_{r_2} < s$ or $a_{r_1} < s , a_{r_2} \geq s$: only one relay decodes the signal.
\end{enumerate}
In the third case, without loss of generality, assume that the first
relay decodes the signal and the second relay does not decode it, i.e, $a_{r_1} \geq s , a_{r_2} < s$. Hence, $X_{r_1}(t)=\alpha X_1(t)$ and
$X_{r_2}(t)=c_2 Y_{r_2}(t)=c_2 \left(h_{r_2} \gamma X_1(t)+Z_{r_2}(t)
\right)$, where $\alpha^2= P_r$ and $c_2 = \sqrt{\frac{ P_r}{a_{r_2} P_s+1}}$. 
At the destination, we have
\begin{align}
\begin{cases}
Y(t) = h_1 \alpha X_1(t) +  h_2c_2h_{r_2}\gamma X_1(t+1) \\
\qquad \quad +h_2c_2Z_{r_2}(t+1)+Z(t),\\
Y(t+1) = -h_1 \alpha X_1^*(t+1) + h_2c_2h_{r_2}^*\gamma X_1^*(t) \\
\qquad \quad +h_2c_2Z_{r_2}^*(t)+Z(t+1).
\end{cases}
\end{align}
After compensating the phase of $h_{r_2}$ into $h_{r_2}^*$ in time slot $t+1$, we get
\begin{align} \label{DAF-single-matrix-decoder}
 \begin{bmatrix}
  Y(t) \\
-Y^*(t+1)
 \end{bmatrix}
=
\begin{bmatrix}
 h_{1} \alpha & h_{r_2}h_{2}c_2\gamma \\
-h_{r_2}^*h_{2}^*c_2\gamma & h_{1}^*\alpha
\end{bmatrix}
\begin{bmatrix}
  X(t) \\
X(t+1)
 \end{bmatrix}
 \nonumber \\
+
\begin{bmatrix}
 c_2 h_{2}Z_{r_2}(t+1) +Z(t)    \\
-c_2 h_{2}^* Z_{r_2}(t) -Z^*(t+1)
\end{bmatrix}.
\end{align}
Multiplying $\begin{bmatrix}
 h_{1}\alpha & h_{r_2}h_{2}c_2\gamma \\
-h_{r_2}^*h_{2}^*c_2\gamma & h_{1}^*\alpha
\end{bmatrix}^\dag
$ to both sides of \cref{DAF-single-matrix-decoder}, two channels are parallelized and the source-destination instantaneous mutual information is
\begin{align}
 \mathcal I\left(X;Y\right) = \ln \left( 1+\frac{\left( |h_{1}|^2\alpha^2+|h_{r_2}h_{2}|^2 c_2^2 \gamma^2 \right) }{1+|h_{2}|^2c_2^2} \right).
\end{align}
A comparison of this method and the DF scheme reveals that if $a_{r_2}>\frac{
P_r}{ P_s}a_1$, then DAF outperforms DF, otherwise, we
switch to DF, that is the second relay becomes silent.
Since the relays do not know the value of $a_1$, they use its expected value. As a result, the amplification coefficient of DAF can be written as $c_\ell = \sqrt{\frac{\mathcal U\left( a_{r_\ell}-\frac{ P_r}{ P_s} \right)  P_r}{a_{r_\ell} P_s+1}}$.
It can be shown that the maximum throughput of this scheme is given by the following proposition.

\begin{proposition}
The maximum throughput of the proposed hybrid decode-amplify-forward relaying is given by
\begin{align} \label{pout_DAF}
& \mathcal R_{DA,s}^m{=}\max_{s} \! \left[\! \left(\! 2e^s \!+\! 2e^{-\frac{P_r}{P_s}}\!+\! s\frac{P_s}{P_r}\!-\! 2e^{s-\frac{P_r}{P_s}}\!-\! 1 \! \right)\! e^{-s\left(2+\frac{P_s}{P_r} \right)} \right. \nonumber \\
&\!+\! \left(e^{-a_{th}} \overline F_{\mathrm a_{AF,2}}(s)\!+\! \left( 1\!-\!e^{-a_{th}} \right)\! \overline F_{\mathrm a_{AF,1}}(s) \right)\! e^{-a_{th}} \!\left( 1\!-\!e^{-s} \right)^2 \nonumber \\
& +\left.  2 e^{-\left( s+\frac{P_r}{P_s} \right)} \left( 1-e^{-s} \right) \overline F_{\mathrm a_{DAF}}\left(s \frac{{P}_{s}}{{P}_{r}} \right)  \right] \ln(1+{P}_{s} s),
\end{align}
where $a_{DAF}=\frac{a_1+a_{r_2}P_s\left(a_1+a_2 \right)}{1+a_{r_2}P_s+a_2 P_r}$, $a_{th}=\frac{P_r}{1+P_s+P_r}$, and $a_{AF,2}$ and $a_{AF,1}$ are from \ref{AF-equivalent-2-channel-gain} and \ref{AF-equivalent-1-channel-gain}, respectively.

\end{proposition}

\subsection{Maximum Finite-Layer Expected-Rate}

Since continuous-layer coding for DAF relaying can not be directly solved by variations methods, we choose a finite-layer code and proceed as follows.
In the finite-layer broadcast approach, the source transmits a $K$ layer
code $X = \sum_{i=1}^K \gamma_{i} X_{i}$ to the relays, where
$\gamma_{i}^2$ represents the power allocated to the $i$'th layer with
rate 
\begin{align} \label{rate_layer_2}
R_{i}=\ln \left( 1+\frac{\gamma_{i}^2 s_i}{1+\sum_{j=i+1}^K \gamma_{j}^2 s_i} \right)
\end{align}
Each relay decodes its received signal from the
first layer up to the layer that its backward channel conditions
allow and forwards them to the destination.
Afterwards, each relay amplifies and
forwards the remaining undecoded layers.

Suppose that the first and second relays allocate portions $\xi P_r$ and $\zeta P_r$ of their power to the
decoded layers, respectively. 
Also, assume that the first and second relays respectively decode $M$ and $N$ layers out of the $K$ transmitted layers. Without loss of
generality, assume $M \geq N$. Denote by $\alpha_i^2$ (resp.\ $\beta_i^2$) the power allocated to the $i$'th layer at the first (resp.\ second) relay. 
The amplifying coefficients are
$c_1=\sqrt{\frac{\bar \xi P_r}{a_{r_1}\sum_{i=M+1}^K
\gamma_i^2 +1}}$ for the first relay and $c_2=\sqrt{\frac{\bar \zeta
P_r}{a_{r_2}\sum_{i=N+1}^K \gamma_i^2 +1}}$ for the second
relay. Let us define $\alpha_i \Def h_{r_1} c_1 \gamma_i$
for $i=M+1,...,K$ and 
$\beta_i \Def h_{r_2} c_2 \gamma_i$
for $i=N+1,...,K$.
The coding scheme is as follows.
At time $t$, the first relay sends $\sum_{i=1}^K \alpha_{i} X_{i} (t)$ while the other relay sends $\sum_{i=1}^K \beta_{i} X_{i} (t+1)$. At time $t+1$, the first and the second relays send $\sum_{i=1}^K - \alpha_{i} X_{i}^{*} (t+1)$ and $\sum_{i=1}^K \beta_{i} X_{i}^{*} (t)$ with compensating the phases of $h_{r_1}$ and $h_{r_2}$ into $h_{r_1}^*$ and $h_{r_2}^*$, respectively. 

The received signal at the destination is
\begin{align} \label{alam1}
\begin{cases}
Y(t){=}h_{1}\! \sum_{i=1}^K\alpha_{i}X_{i}(t)\!+\! h_{2}\! \sum_{i=1}^K \beta_{i}X_{i} (t+1)\! \\
\qquad \quad +h_1 c_1 Z_{r_1}(t) + h_2 c_2 Z_{r_2}(t+1) +\! Z(t), \\
Y(t+1){=}\! -\! h_{1}\! \sum_{i=1}^K\alpha_{i}^*X_{i}^{*}(t+1)\!+\! h_{2}\! \sum_{i=1}^K \beta_{i}^*X_{i}^{*} (t)\! \\
\qquad \quad -h_1 c_1 Z_{r_1}^*(t+1)+h_2 c_2 Z_{r_2}^*(t)+\! Z(t+1).
\end{cases} 
\end{align}
One may express a matrix representation for \cref{alam1} as 
\begin{align} \label{alam4}
 \begin{bmatrix}
  Y(t) \\
-Y^*(t+1)
\end{bmatrix}
{=}
\sum_{i=1}^K
\!
\begin{bmatrix}
 h_{1}\alpha_{i} & \!\!\! h_{2}\beta_{i} \\
-h_{2}^{*}\beta_{i} & \!\!\! h_{1}^{*}\alpha_{i}
\end{bmatrix}
\!\!
 \begin{bmatrix}
  X_{i}(t) \\
X_{i}(t+1)
\end{bmatrix}
\!\!
+ \nonumber \\
\!\!
 \begin{bmatrix}
  h_1 c_1 Z_{r_1}(t) + h_2 c_2 Z_{r_2}(t+1) +\! Z(t) \\
h_1^* c_1 Z_{r_1}(t+1)-h_2^* c_2 Z_{r_2}(t)-\! Z^*(t+1)
\end{bmatrix} .
\end{align}

The destination starts decoding the code layers in order, from the first layer up to the highest layer that is decodable. To decode the $i$'th layer, after decoding the first $i-1$ layers, the channels are separated into two parallel channels by multiplying both sides of \cref{alam4} by $\begin{bmatrix}
 h_{1}\alpha_{i} & \!\!\! h_{2}\beta_{i} \\
-h_{2}^{*}\beta_{i} & \!\!\! h_{1}^{*}\alpha_{i}
\end{bmatrix}^\dag$. Therefore,
\begin{align}
 &\begin{bmatrix}
  \tilde{Y}(t) \\
\tilde{Y}(t+1)
 \end{bmatrix}
= 
\begin{bmatrix}
 a_1\alpha_{i}^{2}+a_2\beta_{i}^{2} & 0 \\
0 & a_1\alpha_{i}^{2}+a_2\beta_{i}^{2}
\end{bmatrix}
\begin{bmatrix}
  X_{i}(t) \\
X_{i}(t+1)
 \end{bmatrix} +
\nonumber \\
&
\!\!\! \sum_{j=i+1}^K \!\! 
\begin{bmatrix}
 h_{1}^{*}\alpha_{i} & -h_{2}\beta_{i} \\
h_{2}^{*}\beta_{i} & h_{1}\alpha_{i}
\end{bmatrix} \!\!
\begin{bmatrix}
 h_{1}\alpha_{j} & h_{2}\beta_{j} \\
-h_{2}^{*}\beta_{j} & h_{1}^{*}\alpha_{j}
\end{bmatrix} \!\!
\begin{bmatrix}
  X_{j}(t) \\
X_{j}(t+1)
 \end{bmatrix} \!\!
 \! +\!
\begin{bmatrix}
  \tilde Z(t) \\
\tilde Z(t+1)
 \end{bmatrix}.
\end{align}
$\tilde Z(t)$ and $\tilde Z(t+1)$ are two independent i.i.d.\ AWGN, each with power $\left(a_1  \alpha_i^2+a_2 \beta_i^2\right)\left(1+a_1 c_1^2+a_2 c_2^2  \right)$.


The interference power caused by upper layers while decoding the $i$'th layer is
\begin{align} \label{cdf_layer_DF_ml_1_2}
I_i&=
\sum_{j=i+1}^K \left( \left( a_1 \alpha_i \alpha_j + a_2 \beta_i \beta_j \right)^2 + a_1 a_2 \left( \alpha_i \beta_j - \alpha_j \beta_i  \right)^2 \right) \nonumber \\
&=\left( a_1 \alpha_i^2 +a_2 \beta_i^2 \right) \sum_{j=i+1}^K \left( a_1 \alpha_j^2 +a_2 \beta_j^2 \right).
\end{align}
Thus, the probability that the $i$'th layer can be correctly decoded at the destination is
{\small \begin{align}
 \mathcal P_i{=}
\Pr \! \left\{\!\! \frac{a_1\alpha_{i}^2+a_2\beta_{i}^2}{1\!+\! a_1 c_1^2\! +\! a_2 c_2^2\! +\!\! \sum_{j=i+1}^K \!\! \left( a_1\alpha_{j}^2+a_2\beta_{j}^2 \right)}\!\! \geq \!\! \frac{\gamma_{i}^2 s_i}{1\!\! +\!\! \sum_{j=i+1}^K \gamma_{j}^2 s_i}\!\! \right\},
\end{align} }
Hence, the expected-rate at the destination using this scheme can be written as 
\begin{align}
 \mathcal R_{DA,f}= 
\sum_{i=0}^K \mathcal P_i \ln \left( 1+\frac{\gamma_{i}^2 s_i}{1+\sum_{j=i+1}^K \gamma_{j}^2 s_i} \right).
\end{align}
To summarize, we have shown the following.
\begin{theorem}
The maximum achievable expected-rate in the proposed DAF relaying is given by
\begin{align}
\mathcal R_{DA,f}^m=\max_{\begin{subarray}{c}
\xi, \zeta, s_i,\gamma_i,
\alpha_i,
\beta_i
\end{subarray}} \sum_{l=0}^K \mathcal P_i \ln \left( 1+\frac{\gamma_{i}^2 s_i}{1+\sum_{j=i+1}^K \gamma_{j}^2 s_i} \right),
\end{align}
where
{\small \begin{align}
 \mathcal P_i{=}
\Pr \! \left\{\!\! \frac{a_1\alpha_{i}^2+a_2\beta_{i}^2}{1\!+\! a_1 c_1^2\! +\! a_2 c_2^2\! +\!\! \sum_{j=i+1}^K \!\! \left( a_1\alpha_{j}^2+a_2\beta_{j}^2 \right)}\!\! \geq \!\! \frac{\gamma_{i}^2 s_i}{1\!\! +\!\! \sum_{j=i+1}^K \gamma_{j}^2 s_i}\!\! \right\},
\end{align} }
and $\alpha_i=\sqrt{\frac{\overline \xi
P_r}{a_{r_1}\sum_{i=M+1}^K \gamma_i^2 +1}}\gamma_i$, $i=M+1,...,K$, 
and $\beta_i=\sqrt{\frac{\overline \zeta
P_r}{a_{r_2}\sum_{i=N+1}^K \gamma_i^2 +1}}\gamma_i$, $i=N+1,...,K$. 
The power constraints are $\sum_{i=1}^K \gamma_{i}^2= P_{s}$, $\sum_{i=1}^M \alpha_{i}^2=\xi  P_r$, and $\sum_{i=1}^N \beta_{i}^2=\zeta P_{r}$. Similar to DF scenario, $\left(\alpha_1, \alpha_2,\dots,\alpha_M,\xi  \right)$ and $\left(\beta_1, \beta_2,\dots,\beta_N,\zeta  \right)$ are optimized separately.
\end{theorem}




\section{Compress-Forward Relays} \label{CF-multi-layer}

In CF relaying, the relays quantize their received signals using an optimal Gaussian quantizer with minimum mean-square error (MSE) criterion \cite{thomas2006}, and then forward the quantized signals. With respect to the correlation between the relays signals, Wyner-Ziv compression method \cite{wyner} is applied.
In this scheme, the relays do not decode the signal and hence, the latency and complexity is lower in comparison with DF and DAF.
Also, the relays do not need to access the source codebook; however, the source-relay channel gains must be available at the destination. 


Denote by $q_{r_1}$ and $q_{r_2}$ the quantized signals at the first and second relays, respectively. One can write the following equations on $q_{r_\ell},~\ell=1,2$,
\begin{align}
 Y_{r_\ell}=q_{r_\ell}+n_{r_\ell},
\end{align}
and 
\begin{align}
 q_{r_\ell}=\theta_{\ell} Y_{r_\ell}+\tilde n_{r_\ell},
\end{align}
where $n_{r_\ell} \sim \mathcal{CN}(0,D_\ell)$ and $\tilde n_{r_\ell} \sim \mathcal{CN}(0,\theta_{\ell} D_\ell)$ are the equivalent quantization noises independent of $q_{r_\ell}$, $\theta_{\ell} \Def 1-\frac{D_\ell}{1+a_{r_\ell}P_{s}}$, and $D_\ell$ is the quantizer distortion at the $\ell$'th relay \cite{berger}.

If the destination decodes $q_{r_1}$ and $q_{r_2}$, and the transmission rate is below $\mathcal I(X;q_{r_1},q_{r_2})$, the signal is successfully decodable. For simplicity, let us assume that the optimum value of the quantizer distortion $D_{\ell}^o$ and the optimum value of the relays rate $R_{r_\ell}^o$ are selected independent of the source-relays channel gains. Hence, with respect to the network symmetry, $D_{1}^o=D_{2}^o$ and $R_{r_1}^o=R_{r_2}^o$, and therefore, they are simply denoted by $D$ and $R_{r}$, respectively. 

To decoded the quantized signals at the destination, based on the multiple-access capacity region \cite{thomas2006} in the second-hop, the following inequalities must be satisfied,
\begin{align} \label{mac_part}
&R_{r}<\mathcal I(X_{r_1};Y|X_{r_2})=\ln \left( a_{1}P_r+1 \right), \nonumber \\
&R_{r}<\mathcal I(X_{r_2};Y|X_{r_1})=\ln \left( a_{2}P_r+1 \right), \nonumber \\
&2R_{r}<\mathcal I(X_{r_1},X_{r_2};Y)=\ln \left( (a_{1}+a_{2})P_r+1 \right).
\end{align}

For lossless compression of the quantized signals, based on the Wyner-ziv rate region \cite{wyner}, we have the following inequalities,
\begin{align} 
&R_{r} \geq \mathcal I(q_{r_1};Y_{r_1}|q_{r_2}), \label{1wz_part} \\
&R_{r} \geq \mathcal I(q_{r_2};Y_{r_2}|q_{r_1}), \label{2wz_part} \\
&2R_{r} \geq \mathcal I(q_{r_1},q_{r_2};Y_{r_1},Y_{r_2}). \label{3wz_part}
\end{align}


In the problem in consideration, \cref{3wz_part} is 
\begin{align} \label{sw_3_cf}
\mathcal I(q_{r_1},q_{r_2};Y_{r_1},Y_{r_2}) &= \ln \left( \frac{\det \mathbf R_{Y_1Y_2}}{\det \mathbf R_{Y_1Y_2|q_{r_1},q_{r_2}}}  \right) \nonumber \\
&= \ln \left( \frac{(a_{r_1}+a_{r_2})P_s+1}{D^2} \right).
\end{align}

In order to derive a closed form expression for \cref{1wz_part,2wz_part}, let us first estate the following lemmas. 

\begin{lemma}

The mutual information between the source signal and the relays quantized signals is given by
\begin{align} \label{I-q-q-X-eq-}
\mathcal I(q_{r_1},q_{r_2};X)=\ln \left( 1+a_{CF} P_s \right),
\end{align}
where,
\begin{align} \label{def-a-cf-eq} 
a_{CF} \Def \frac{a_{r_1}}{1+\frac{\theta_2+D}{\theta_2+1}\frac{D}{\theta_1}}+\frac{a_{r_2}}{1+\frac{\theta_1+D}{\theta_1+1}\frac{D}{\theta_2}}.
\end{align}
\end{lemma}

\begin{proof}
The mutual information between the source signal and the relays quantized signals can be expressed by
\begin{align}
 \mathcal I(q_{r_1},q_{r_2};X)=\ln \left( \frac{\det \mathbf R_{q_{r_1}q_{r_2}}}{\det \mathbf R_{q_{r_1}q_{r_2}|X}}  \right),
\end{align}
where
\begin{align}
 \det \mathbf R_{q_{r_1}q_{r_2}}=\theta_{1}^2 \theta_{2}^2 a_{r_1} P_s+ \theta_{1}^2 \theta_{2} D a_{r_1} P_s + \theta_{1}^2 \theta_{2}^2 a_{r_2} P_s+\theta_{1}^2 \theta_{2}^2 \nonumber \\
+ \theta_{1}^2 \theta_{2} D+ \theta_{1} \theta_{2}^2 a_{r_2} P_s D+ \theta_{1} \theta_{2}^2 D + \theta_{1} \theta_{2} D^2,
\end{align}
and
\begin{align}
\det \mathbf R_{q_{r_1}q_{r_2}|X}=\theta_{1}^2 \theta_{2}^2+ \theta_{1}^2 \theta_{2} D+ \theta_{1} \theta_{2}^2 D+ \theta_{1} \theta_{2} D^2.
\end{align}
Thus,
\begin{align} \label{a-cf-eq-proof}
\mathcal I(X\!;\!q_{r_1},\!q_{r_2}\!)&{=}\!\ln\! \left(\!\! 1\!\!+\!\frac{\!\theta_{1} \theta_{2} a_{r_1} \!\!+\!\theta_{1} D a_{r_1} \!\!+\!\theta_{1} \theta_{2} a_{r_2} \!\!+\!\theta_{2} a_{r_2}  D}{\theta_{1} \theta_{2}+ \theta_{1} D+  \theta_{2} D+  D^2}\!P_s\!  \right)
\nonumber \\
&{=} \ln \Bigg( 1+\left( a_{r_1}\frac{\theta_1\theta_2+\theta_1 D }{\theta_1\theta_2+\theta_1D+\theta_2D+D^2} \right. \nonumber \\
&+ \left. a_{r_2}\frac{\theta_1\theta_2+\theta_2D}{\theta_1\theta_2+\theta_1D+\theta_2D+D^2} \right) P_s \Bigg) \nonumber \\
&{=} \ln\left( 1+ \left( \frac{a_{r_1}}{1+\frac{\theta_2D+D^2}{\theta_1\theta_2+\theta_1}}+\frac{a_{r_2}}{1+\frac{\theta_1D+D^2}{\theta_1\theta_2+\theta_2}}   \right) P_s \right) \nonumber \\
&{=} \ln\left( 1\!+\! \left( \frac{a_{r_1}}{1\!+\!\frac{\theta_2+D}{\theta_2+1}\frac{D}{\theta_1}}+\frac{a_{r_2}}{1+\frac{\theta_1+D}{\theta_1+1}\frac{D}{\theta_2}}  \right) P_s \right).
\end{align}

\Cref{def-a-cf-eq} together with \cref{a-cf-eq-proof} results.
\end{proof}

\begin{lemma} \label{lemma-cf-I-q1-yr1-condition-q2}
In the problem of interest, we have
\begin{align} \label{sw_2_cf}
\mathcal I(q_{r_1};Y_{r_1}|q_{r_2}){=}\ln \left( 1\!+\!a_{CF} P_s \right)\!+\!\ln \left( \frac{\left(\theta_1\!+\!D\right)\left(\theta_2\!+\!D\right)}{D\left( 1+a_{r_2}P_s \right)}  \right).
\end{align}
\end{lemma}

\begin{proof}

\begin{align}  \label{sw_2_cf-proof}
\mathcal I(q_{r_1};Y_{r_1}|q_{r_2})&=\mathcal I(q_{r_1};X,Y_{r_1}|q_{r_2})-\mathcal I(q_{r_1};X|Y_{r_1},q_{r_2}) \nonumber \\
&\stackrel{(a)}{=} \mathcal I(q_{r_1};X,Y_{r_1}|q_{r_2}) \nonumber \\
&=\mathcal I(q_{r_1},q_{r_2};X,Y_{r_1})-\mathcal I(q_{r_2};X,Y_{r_1})  \nonumber \\
&=\mathcal I(q_{r_1},q_{r_2};X,Y_{r_1})-\mathcal I(q_{r_2};Y_{r_1}|X) \nonumber \\
&-\mathcal I(q_{r_2};X) \nonumber \\
&\stackrel{(b)}{=} \mathcal I(q_{r_1},q_{r_2};X,Y_{r_1})-\mathcal I(q_{r_2};X)\nonumber\\
&=\mathcal I(q_{r_1},q_{r_2};X)+\mathcal I(q_{r_1},q_{r_2};Y_{r_1}|X)\nonumber\\
&-\mathcal I(q_{r_2};X) \nonumber \\
&=\mathcal I(q_{r_1},q_{r_2};X)+\mathcal I(q_{r_1};Y_{r_1}|X)\nonumber\\
&+\mathcal I(q_{r_2};Y_{r_1}|q_{r_1},X)-\mathcal I(q_{r_2};X) \nonumber \\
&\stackrel{(c)}{{=}}
\mathcal I(q_{r_1},q_{r_2};\!X)\!+\!\mathcal I(q_{r_1};Y_{r_1}|X)\!-\!\mathcal I(q_{r_2};\!X) \nonumber \\
&=\mathcal I(q_{r_1},q_{r_2};X)\!+\!\mathcal H(q_{r_1}|X)\! \nonumber \\
&-\!\mathcal H(q_{r_1}|Y_{r_1},X)\!-\!\mathcal H(q_{r_2})\!+\!\mathcal H(q_{r_2}|X) \nonumber \\
&\stackrel{(d)}{=}\mathcal I(q_{r_1},q_{r_2};X)\!+\!\mathcal H(q_{r_1}|X)\! \nonumber \\
&-\!\mathcal H(q_{r_1}|Y_{r_1})\!-\!\mathcal H(q_{r_2})\!+\!\mathcal H(q_{r_2}|X) \nonumber \\
&=\ln \left( 1+a_{CF} P_s \right)\!+\!\ln \left(1+\frac{\theta_{1}}{D} \right)\! \nonumber \\
&-\! \ln \left( 1+\frac{\theta_{2}a_{r_2}P_s}{\theta_{2}+D} \right) \nonumber \\
&=\ln \left( 1+a_{CF} P_s \right)\!\nonumber \\
&+\!\ln \left( \frac{\left(\theta_1+D\right)\left(\theta_2+D\right)}{D\left( D+\theta_2 \left(1+a_{r_2}P_s \right) \right)}  \right)\nonumber\\
&{=}\ln \left( 1\!+\!a_{CF} P_s \right)\!+\!\ln \left( \frac{\left(\theta_1\!+\!D\right)\left(\theta_2\!+\!D\right)}{D\left( 1+a_{r_2}P_s \right)}  \right).
\end{align}
$(a)$ and $(d)$ follow from the fact that $X\longmapsto Y_{r_1}\longmapsto q_{r_1}$ is a Markov chain, and hence $\mathcal I(q_{r_1};X|Y_{r_1},q_{r_2})=0$ and $\mathcal H(q_{r_1}|Y_{r_1},X) = \mathcal H(q_{r_1}|Y_{r_1})$. 
$(b)$ and $(c)$ follow from $\mathcal I(q_{r_2};Y_{r_1}|X)=0$ and $\mathcal I(q_{r_2};Y_{r_1}|q_{r_1},X)=0$, respectively, with respect to the Markov chain $q_{r_2}\longmapsto X\longmapsto Y_{r_1}$.
\end{proof}
With respect to the network symmetry and based on \cref{lemma-cf-I-q1-yr1-condition-q2}, one can express
\begin{align} \label{sw_1_cf}
\mathcal I(\!q_{r_2};\!Y_{r_2}|q_{r_1}\!){=}\ln \left( 1\!+\!a_{CF} P_s \right)\!+\!\ln \left( \frac{\left(\theta_1\!+\!D\right)\left(\theta_2\!+\!D\right)}{D\left( 1+a_{r_1}P_s \right)}  \right).
\end{align} 

In order to have a successful transmission, the destination must first decode the relays signals and then $X$. 
From \cref{1wz_part,2wz_part,3wz_part,mac_part,sw_3_cf,sw_2_cf,sw_1_cf}, to decode the relays signals at the detination, the following inequalities must be satisfied. 
\begin{align} \label{CF-P-C-decode-region}
&\ln \!\left(\! 1\!+\!a_{CF} P_s \!\right)\!+\!\ln\! \left(\! \frac{\left(\theta_1\!+\!D\right)\left(\theta_2\!+\!D\right)}{D\left( 1\!+\!a_{r_2}P_s \right)}  \right) \!\leq\! R_r\!<\!\ln \left(1\!+\! a_{1}P_r \!\right), \nonumber \\
&\ln \!\left(\! 1\!+\!a_{CF} P_s \!\right)\!+\!\ln\! \left(\! \frac{\left(\theta_1\!+\!D\right)\left(\theta_2\!+\!D\right)}{D\left( 1\!+\!a_{r_1}P_s \right)}  \right) \!\leq \!R_r\!<\!\ln \left(1\!+\! a_{2}P_r \!\right), \nonumber \\
&\ln\left( \frac{\left( a_{r_1}+a_{r_2} \right)P_s+1}{D^2} \right) \leq 2R_r<\ln \left( (a_{1}+a_{2})P_r+1 \right).
\end{align}

Therefore, the probability of decoding the relays signals at the destination is expressed as follows,

\begin{align} \label{CF-P-C-decode-=rob-quantizeds}
 \mathcal P_C=\Pr \Bigg\{ 
\max \left\{ \ln\left( \frac{\sqrt{\left( a_{r_1}+a_{r_2} \right)P_s+1}}{D} \right),  \right. \nonumber \\
 \left. \ln \left( 1+a_{CF} P_s \right) \!\!+\!\ln \left( \frac{\left( \theta_1+D\right) \left( \theta_2+D \right)}{D \left( 1+a_{r_{min}}P_s \right)}  \right) \right\} \nonumber \\
 <R_{r}< \nonumber \\
 \min \left\{ \ln\left( \sqrt{\left( a_{1}+a_{2} \right)P_r+1} \right), \ln \left( a_{min}P_r+1 \right) \right\} 
\Bigg\},
\end{align}
where $a_{r_{min}}\Def \min\left\{ a_{r_1},a_{r_2} \right\}$ and $a_{min}\Def \min\left\{ a_1,a_2 \right\}$.

After decoding the relays signals at the destination, the source signal is decoded subject to
\begin{align}
R \leq  \mathcal I(q_{r_1},q_{r_2};X) = \ln \left( 1+a_{CF} P_s \right),
\end{align}
where $R=\ln \left( 1+P_s s \right)$ is the source transmission rate.

To summarize, we have shown the following.
\begin{theorem}
The maximum throughput in the proposed CF scheme is expressed by
\begin{align}
 \mathcal R_{C,s}^m=\max_{s, D, R_{r}} \mathcal P_C \overline F_{\mathrm a_{CF}}(s) \ln \left( 1+P_s s \right),
\end{align}
where $a_{CF}$ and $\mathcal P_C$ are given by \cref{def-a-cf-eq,CF-P-C-decode-=rob-quantizeds}, respectively.

Analogously, \cref{infinite-layer_CF} yields the maximum continuous-layer expected-rate in this scheme.
\begin{align} \label{infinite-layer_CF}
 \mathcal R_{C,c}^m=\max_{\begin{subarray}{c}
D, R_{r}
\end{subarray}} \mathcal P_C \int_{s_{0}}^{s_{1}} \overline F_{\mathrm a_{CF}}(s) \left( \frac{2}{s}+\frac{ f_{\mathrm a_{CF}}^{\prime}(s)}{f_{\mathrm a_{CF}}(s)} \right)\ud s.
\end{align}
The integration limits are the solutions to $\overline F_{\mathrm a_{CF}}(s_0)=s_0 \left( 1+P_s s_0  \right) f_{\mathrm a_{CF}}(s_0)$ and $\overline F_{\mathrm a_{CF}}(s_1)=s_1 f_{\mathrm a_{CF}}(s_1)$, respectively.
\end{theorem}

It turns out from the numerical results that the proposed CF scheme outperforms DAF and consequently, DF and AF, when the relay power to the source power ratio is higher than a threshold. 

\begin{remark}
If $P_{s}\to \infty$, \cref{def-a-cf-eq} is simplified to $a_{CF} \approx \frac{a_{r_1}+a_{r_2}}{1+\frac{D(D+1)}{2}}$.
%
If $P_{r}\to \infty$, then $\mathcal P_C \to 1$ and $a_{CF} \approx a_{r_1}+a_{r_2}$.
In this high SNR asymptote at the relays, \cref{infinite-layer_CF} meets the cutset-bound of
\cref{up-general-cutset-proposition} in  \cref{cutset_general_section}
, and is optimum. 
\end{remark}


\section{Upper-Bounds} \label{upper-bound-multicast-relay}

\subsection{Cutset Bound} \label{cutset_general_section}

The network cutset bound is the minimum of the maximum throughput and maximum expected-rate of the first-hop and the second-hop which lend itself to a closed form expression.
The first-hop cutset is equivalent to a point-to-point single-input multiple-output (SIMO) channel with two receive antennas. 
The second-hop is equivalent to a multiple-input single-output (MISO) channel with two transmit antennas. 
The throughput cutset bound is the minimum of the maximum throughput in these two cutsets, that is
\begin{align} \label{cutset-throughput}
\mathcal R_{CS,s}^m = \max_s e^{-s} (1+s) \ln \left( 1+ P s \right),
\end{align}
where $P \Def \min \left\{ P_s,P_r  \right\}$.

Similarly, the maximum expected-rate of the diamond channel is upper-bounded by the minimum of the maximum expected-rates of those two cutsets, which is summarized below. 

\begin{proposition} \label{up-general-cutset-proposition}
In the diamond channel, the cutset bound on the maximum expected-rate is specified by
 \begin{align}
\mathcal R_{CS,c}^{m}{=}3\emph{E}_1(s_0)\! -\! 3\emph{E}_1(s_1)\! -\! (s_0\! -\! 1)e^{-s_0}\!+\! (s_1\!-\! 1)e^{-s_1},
\end{align}
where $s_1=\frac{1+\sqrt{5}}{2}$, and $s_0 = \sqrt[3]{\sqrt{ A^2-B^3}+A}  
 +\frac{B}{\sqrt[3]{\sqrt{ A^2-B^3}+A}}-\frac{1}{3 P}$
 with $P=\min\left\{P_s,P_r\right\}$,
 $A=\frac{1}{2 P}-\frac{1}{6 P^2}-\frac{1}{27 P^3}$, and $B=\frac{1}{3 P}+\frac{1}{9 P^2}$.
\end{proposition}

\begin{proof}

%

As showed in \cite{shamai}, the maximum continuous-layer expected-rate is given by 
\begin{align} \label{infinite-layer-general-me}
\mathcal R_{c}^m = \max_{I(s)} \int_{0}^{\infty} \overline F_{\mathrm a}(s) \frac{-s I^{\prime}(s)}{1+sI(s)} \ud s.
\end{align}
Noting $\overline F_{\mathrm a}(s)= \left( 1+s  \right)e^{-s}$ based on \cite{zamani2012maximum}, we have
\begin{align} \label{infinite-layer-general-shamai}
\mathcal R_{c}^m = \max_{I(s)} \int_{0}^{\infty} \frac{-s\left(1+s\right)e^{- s} I^{\prime}(s)}{1+sI(s)} \ud s.
\end{align}
The optimization solution to \cref{infinite-layer-general-shamai} with respect to $I(s)$ under the total power constraint $P = \min \left\{ P_s,P_r  \right\}$ is found using variation methods \cite{calculus}. By solving the corresponding E$\ddot{\text{u}}$ler equation \cite{calculus}, we come up with the final solution as follows,
\begin{align} \label{infinite-layer-miso-integral}
\mathcal R_{c}^m = \int_{s_0}^{s_1}e^{-s}\left(1+s\right) \left( \frac{3}{s} -1  \right)  \ud s,
\end{align}
where boundaries $s_0$ and $s_1$ are the solutions to $P s_0^3+s_0^2-s_0-1=0$ and $s_1^2-s_1-1=0$, respectively.
Therefore, $s_1=\frac{1+\sqrt{5}}{2}$, and $s_0 = \sqrt[3]{\sqrt{ A^2-B^3}+A}  
 +\frac{B}{\sqrt[3]{\sqrt{ A^2-B^3}+A}}-\frac{1}{3 P}$
 with  $A=\frac{1}{2 P}-\frac{1}{6 P^2}-\frac{1}{27 P^3}$, and $B=\frac{1}{3 P}+\frac{1}{9 P^2}$.
The indefinite integral (antiderivative) of \cref{infinite-layer-miso-integral} is 
\begin{align}
\int e^{-s}\left(1+s\right) \left( \frac{3}{s} -1  \right)  \ud s=(s- 1)e^{-s}- 3\emph{E}_1(s).
\end{align}
Applying the integration limits completes the proof.

\end{proof}

\subsection{Relay-Cooperation (RC) Bound} \label{rc-section-upper-bound}

Here, a tighter upper-bound based on a full-cooperation between the relays is proposed. 
Let us define an \emph{upper-bound model} by considering a full cooperation and power cooperation between the relays in the problem of interest (see \cref{upper_bound_model}). The \emph{upper-bound model} is equivalent to a dual-hop single-relay channel with two antennas at the relay. 
\begin{figure}
\centering
\includegraphics[height=1cm, width=.4\linewidth]{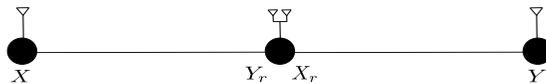}
\caption{\emph{Upper-bound model}.}
\label{upper_bound_model}
\end{figure}
The following presents the throughput of this \emph{upper-bound model}.
\begin{proposition} \label{upper-bound-one-layer-general-rc}
the maximum throughput in the above upper-bound model is given by
\begin{align} \label{throughput-upper-bound}
\mathcal R_{RC,s}^m {=} \max_s \left( 1+s  \right)\!\! \left( 1+s\frac{P_s}{P_r}  \right)\!\! e^{-s\left( 1+\frac{P_s}{P_r} \right)}\! \ln \left( 1+P_s s  \right).
\end{align}
\end{proposition}
 
\begin{proof}
The optimum relaying strategy for dual-hop single-relay channels is DF.
In the same general outline to the proof of \cref{DF-single-theorem}, let $\overline F\left(s\right)$ denote
\begin{align} \label{throughput-upper-bound13}
\overline F\left(s\right) = \Pr\left\{ a_r>s \right\}\Pr\left\{ a>\frac{P_s}{P_r}s \right\},
\end{align}
where $a=\frac{1}{P_r} \vec h\mathbf Q\vec h^\dag$ and $\vec h = \left[\begin{matrix} h_1 & h_2 \end{matrix}\right]$.
The maximum throughput in general can be expressed as
\begin{align} \label{throughput-upper-bound14}
\mathcal R_{RC,s} = \max_s \overline F(s) \ln(1+P_{s} s).
\end{align}
Analogously to the proof of \cref{DF-single-theorem}, we can restrict our attention to $\rho=0$ or $\rho=1$, where $\rho$ is the correlation coefficient between the signals transmitted from two relay antennas.
To prove by contradiction, first we assume that $\rho^o=1$; next we shall show that $\overline F_{\rho=0} (s^o)>\overline F_{\rho=1} (s^o)$, which implies a contradiction and concludes $\rho^o=0$.
Defining 
\begin{align}
&g(s)\Def \ln \left( 1+P_s s \right)^{\frac{1+P_s s }{P_s}}, \\
&r(s)\Def \frac{\overline F_{\rho=1}(s)}{f_{\rho=1}(s)},
\end{align}
\cref{two_dist_single,diff_compair_trg} hold. 

Noting 
\begin{align}
\overline F_{\rho=1}(s)=\left( 1+s  \right) e^{-s\left( 1+\frac{P_s}{2P_r} \right)},
\end{align}
we have
\begin{align} 
r(s)=\frac{1+s}{\left(1+s  \right)\left(1+\frac{P_s}{2P_r}  \right)-1}.
\end{align}
It can be shown that 
\begin{align} \label{throughput-upper-bound17} 
r(s) < s, ~~~\forall s\geq s_t,
\end{align}  
where
\begin{align} \label{throughput-upper-bound18} 
s_t \Def 
\frac{\sqrt{P_s^2+4P_s P_r+20 P_r^2}-P_s+2P_r}{2P_s+4P_r}.
\end{align}  
Hence, \cref{lhs_diff_miso_single_Rayleigh6,diff_miso_single_Rayleigh1} still hold by redefining $r(s)$ and $s_t$ as above.

As $\mathcal R(s)$ is a continuous function, one can conclude that $0< s^o < s_t$. Noting 
\begin{align}
s_t < s_c=- \left( 2\mathcal W_{-1} \left(\frac{-1}{2\sqrt{e}}\right)+1\right) \frac{P_r}{P_s} \approx 2.5129\frac{P_r}{P_s}
\end{align}
and 
\begin{align}
\overline F_{\rho=0}(s) > \overline F_{\rho=1}(s), ~~ \forall s<s_c
\end{align}
yields $\overline F_{\rho=0}(s^o) > \overline F_{\rho=1}(s^o)$ and thereby, $\rho^o=0$ and $a=a_1+a_2$. Substituting the channel gain CDFs in \cref{throughput-upper-bound14,throughput-upper-bound13}, the maximum throughput of the DF diamond channel is given by \cref{throughput-upper-bound}.

\end{proof}

The highest expected-rate of dual-hop single-relay channels has been studied in \cite{steiner}. 
Here, only the final solution is mentioned as
\begin{align} \label{gen_upp}
\mathcal R_{RC,c}^m{=}&\max_{\begin{subarray}{c}
I_s(a_{r})\\
I_r(a\left|{a_{r}}\right.)
\end{subarray}
}
\! \int\limits_0^\infty \! \int\limits_0^\infty f_{a_{r}}(t)\overline F_{a}(s) \frac{-s I_r^{\prime}(s|{a_{r}}=t)\ud s}{1+s I_r(s|{a_{r}}=t)} \ud t.
\end{align}
The power constraints at the transmitter and the relay are 
\begin{align} \label{ub-power}
I_s(0)=P_s,~ I_r(0|a_r=t)=P_r.
\end{align}
As the maximum transmission rate of the relay can not exceed its successfully decoded rate, the constraint on rate is 
\begin{align} \label{ub-rate-constraint}
\int_0^\infty \frac{s I_r^{\prime}(s|a_{r}=t)\ud s}{1+s I_r(s|a_{r}=t)} = \int_0^t \frac{s I_s^{\prime}(s)\ud s}{1+sI_s(s)},~~ \forall t.
\end{align}
The optimization problem of \cref{gen_upp} can be solved numerically using the algorithm proposed in \cite{vahid}. 

Following a similar outline in the proof of \cref{df-multi-strategy-theorem,upper-bound-one-layer-general-rc}, one can show that the optimum transmission strategy at the relay is to transmit uncorrelated equal power signals from both of the relay antennas at each layer.
Thus, $\overline F_{\mathrm a_{r}}(s)=\overline F_{\mathrm a}(s)=(1+s)e^{-s}$.
Substituting in \cref{gen_upp},  we come up with the upper-bound as follows, which does not lend itself to a closed form formulation.

\begin{proposition}
In parallel relay networks, the maximum expected-rate at the destination is bounded by
\begin{align} \label{gen_upp_rayleigh_2}
\mathcal R_{c}^{rc}{=}&\max_{\begin{subarray}{c}
I_s(a_{r})\\
I_r(a\left|{a_{r}}\right.)
\end{subarray}
} 
\!\! \int\limits_0^\infty \!\! te^{-t}\!\!\! \int\limits_0^\infty \frac{s(s+1) e^{-s} I_r^{\prime}(s|{a_{r}}=t)}{1+sI_r(s|{a_{r}}=t)}\emph{d} s \emph{d}t,
\end{align}
subject to the power and rate constraints \cref{ub-power,ub-rate-constraint}, respectively.
\end{proposition}

\subsection{DF-Upper-Bounds} \label{DF-upper-bound-section}

As pointed out earlier, the continuous-layer coding for DF relaying can not be directly solved by variations methods. Here, two upper-bounds for the maximum continuous-layer expected-rate in DF scheme are obtained. Let us define a \emph{DF-upper-bound model} as a diamond channel with uninformed transmitters, wherein the channel gains of the transmitter-relay links are both $\max\{a_{r_1},a_{r_2}\}$, and that of the relay-destination links are $a_1$ and $a_2$. This channel can be modeled by a dual-hop single-relay channel with the channel gains $a_{r}=\max\{a_{r_1},a_{r_2}\}$ and $a$ for the transmitter-relay link and the relay-destination link, respectively. 
Clearly, the maximum expected-rate of this model yields an upper-bound on the maximum expected-rate of DF relaying.


The optimum relaying strategy in the \emph{DF-upper-bound model} is decode-forward, and is given by \cref{gen_upp}. 
Analogous to \cref{rc-section-upper-bound}, it can be shown that the optimum transmission strategy at the relay is to transmit uncorrelated equal power signals from the relays at each layer.
Hence, substituting $\overline F_{\mathrm a_{r}}(s)=e^{-s} \left( 2-e^{-s} \right)$ and $\overline F_{\mathrm a}(s)=(1+s)e^{-s}$ in \cref{gen_upp},  we come up with the upper-bound as follows, which does not lend itself to a closed form formulation.

\begin{proposition} 
In the DF diamond channel, the maximum expected-rate at the destination is bounded by
\begin{align} \label{gen_upp_rayleigh}
\mathcal R_{cl}^m \!\!=\!\!\!\! &\max_{\begin{subarray}{c}
I_s(a_{r})\\
I_r(a\left|{a_{r}}\right.)
\end{subarray}
} 
\!\!\! 2\!\! \int\limits_{0}^{\infty} \!\!\! e^{-t}\! \left( e^{-t}-1 \right)\!\!\!\! \int\limits_{0}^{\infty} \!\!\! \frac{s(s+1) e^{-s} I_r^{\prime}(s|{a_{r}}=t)}{1+sI_r(s|{a_{r}}=t)}\emph{d} s \emph{d}t,
\end{align}
subject to the power and rate constraints \cref{ub-power,ub-rate-constraint}, respectively.
\end{proposition}

The cutset bound of the \emph{DF-upper-bound model} results in a closed form expression. The results are summarized below.

\begin{proposition} \label{dfup-general-cutset-proposition}
The cutset bound of the {DF-upper-bound model} is specified by $\min\left\{\mathcal R_1, \mathcal R_2 \right\}$, where 
 \begin{align}
&\mathcal R_1 {=} 4\emph{E}_1(s_1\!)\!-\!2\emph{E}_1(2s_1\!)\!+\!e^{-2s_1} \!\!-\! 3e^{\!-s_1} \!\!-\! \ln \! \left( \! 1 \!\!-\!e^{-s_1}\!\right)\!-\!0.1157\!, \nonumber \\
&\mathcal R_2 {=} 3\emph{E}_1(s_2)-(s_2-1)e^{-s_2}-0.1296.
\end{align}
$s_1$ is the solution to $\frac{2-e^{-s_1}}{2s_1\left(1-e^{-s_1} \right)}=1+P_s s_1$ and
$s_2 = \sqrt[3]{\sqrt{ A^2-B^3}+A}  
 +\frac{B}{\sqrt[3]{\sqrt{ A^2-B^3}+A}}-\frac{1}{3 P_r}$
 with $A=\frac{1}{2 P_r}-\frac{1}{6 P_r^2}-\frac{1}{27 P_r^3}$, and $B=\frac{1}{3 P_r}+\frac{1}{9 P_r^2}$.
\end{proposition}

\begin{proof}
The bound on the second hop, i.e., $\mathcal R_{2}$, is a direct result of \cref{up-general-cutset-proposition}. 
%

Noting $\overline F_{\mathrm a_{r}}(s)=e^{-s} \left( 2-e^{-s} \right)$ in the first hop, analogous to the proof of \cref{up-general-cutset-proposition}, we have
\begin{align} \label{infinite-layer-general-shamai-dfup}
\mathcal R_{1} = \max_{I(s)} \int_{0}^{\infty} \frac{-e^{-s} \left( 2-e^{-s} \right) I^{\prime}(s)}{1+sI(s)} \ud s.
\end{align}
The optimization solution to \cref{infinite-layer-general-shamai-dfup} with respect to $I(s)$ under the total power constraint $P_s$ is found by solving the associated E$\ddot{\text{u}}$ler equation \cite{calculus}, which leads to
\begin{align} \label{infinite-layer-miso-integral_2}
\mathcal R_{1} = \int_{s_0}^{s_1} e^{-s} \left( 2-e^{-s} \right) \left( \frac{2}{s} +\frac{2e^{-s}-1}{1-e^{-s}}  \right)  \ud s,
\end{align}
where boundaries $s_0$ and $s_1$ are the solutions to $\frac{2-e^{-s_0}}{2s_0\left(1-e^{-s_0} \right)}=1+P_s s_0$ and $\frac{2-e^{-s_1}}{2s_1\left(1-e^{-s_1} \right)}=1$, respectively.
The indefinite integral (antiderivative) of \cref{infinite-layer-miso-integral_2} is 
\begin{align}
&\int e^{-s} \left( 2-e^{-s} \right) \left( \frac{2}{s} +\frac{2e^{-s}-1}{1-e^{-s}}  \right)  \ud s= \nonumber \\ &-4\text{E}_1(s_1)\!+\!2\text{E}_1(2s_1)\!-\!e^{-2s_1} \!+\! 3e^{-s_1} \!+\! \ln \left(1-e^{-s_1}\right)\!.
\end{align}
Applying the integration limits completes the proof.

\end{proof}


\section{Numerical Results} \label{numerical-results}

The achievable throughput, two-layer expected-rate, and continuous-layer expected-rate in the proposed multi-layer relaying schemes and their upper-bounds are shown respectively in \cref{single-layer-fig,two-layer-fig,continuous-layer-fig} for $P_s=0$ dB and $0 \leq P_r \leq 60$ dB.
Note that the rates are expressed in \emph{nats}.
When $\frac{P_r}{P_s}$, namely \textit{powers ratio}, is less than $25$ dB, DAF is the best scheme. In higher values of the \textit{powers ratio}, CF is the superior. 
AF has the worst performance for $\frac{P_r}{P_s}>10$ dB. 
When the \textit{powers ratio} goes to infinity, CF meets the upper-bounds.

\begin{figure}
\centering
\includegraphics[scale=1]{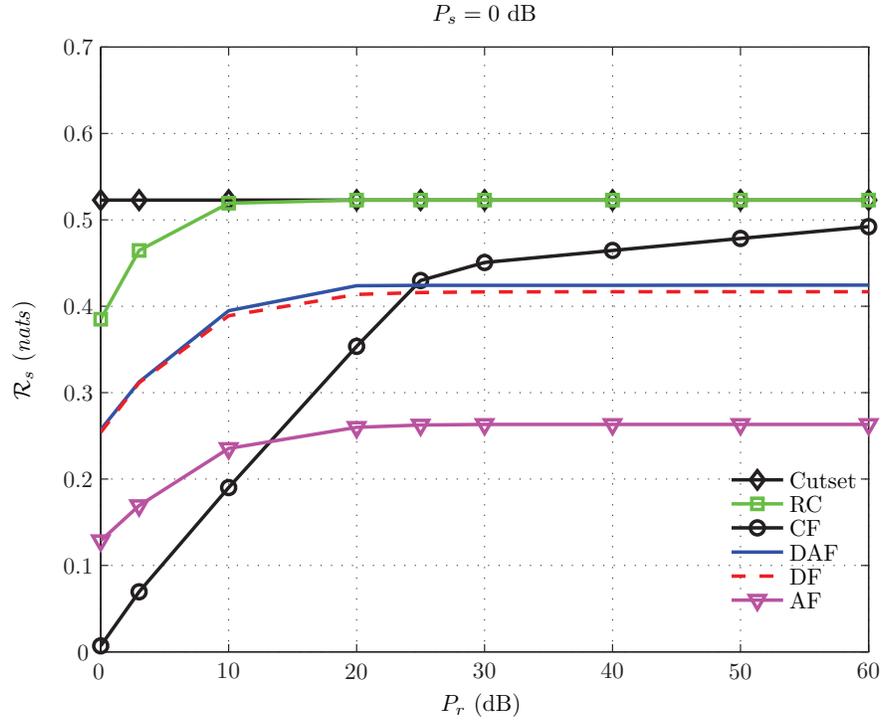}
\caption{Throughput in the diamond channel.}
\label{single-layer-fig}
\end{figure}

\begin{figure}
\centering
\includegraphics[scale=1]{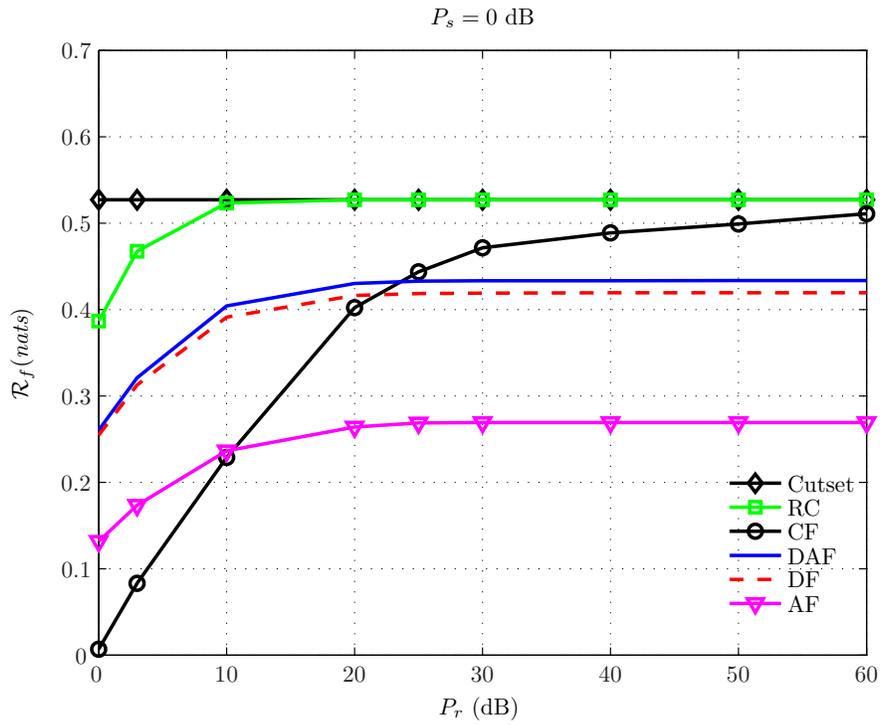}
\caption{Two-layer expected-rate in the diamond channel.}
\label{two-layer-fig}
\end{figure}

\begin{figure}
\centering
\includegraphics[scale=1]{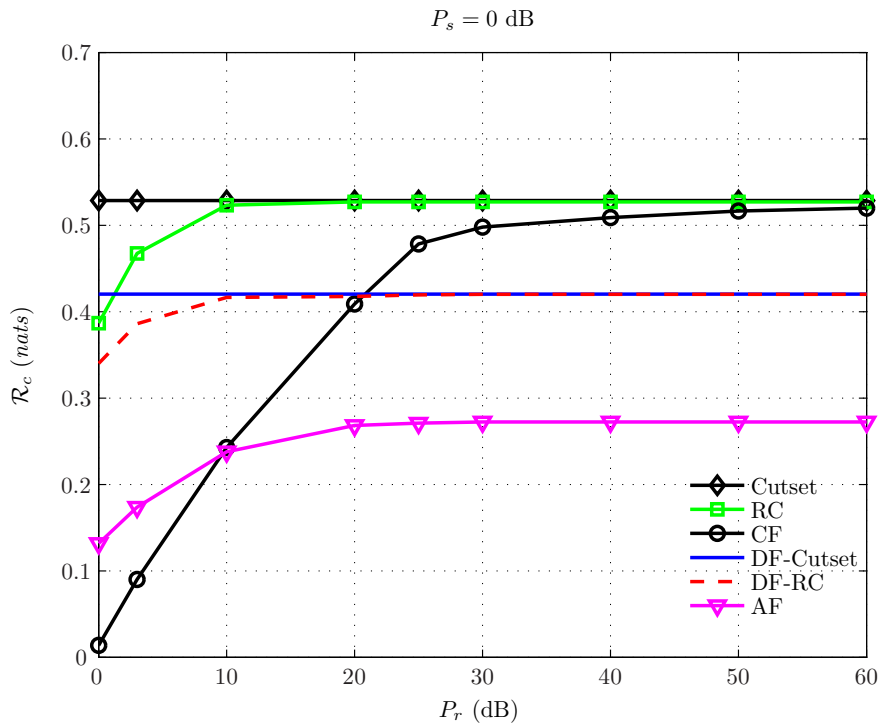}
\caption{Continuous-layer expected-rate in the diamond channel.}
\label{continuous-layer-fig}
\end{figure}

\section{Conclusion} \label{conclusion}

The main goal of the paper is to propose simple, efficient,
and practical relaying schemes to increase the average achievable rate at the destination in dual-hop parallel relay networks with Rayleigh block fading links and uninformed transmitters.
To this end, different relaying schemes, in conjunction with the broadcast approach, were proposed. 
The performance of the proposed schemes were derived and numerically compared with two obtained upper-bounds.




%
%
%


\bibliographystyle{IEEEtran}
\bibliography{bibliography}
\end{document}